\documentclass[11pt]{article}

\usepackage{geometry}
\usepackage[numbers]{natbib}
\usepackage{amsfonts,amsmath,amssymb,amsthm}
\usepackage{graphicx}
\usepackage[colorlinks,citecolor=green,linkcolor=blue,bookmarks=false,hypertexnames=true]{hyperref}
\usepackage[color=yellow]{todonotes}

\newtheorem{proposition}{Proposition}
\newtheorem{remark}{Remark}
\newtheorem{assumption}{Assumption}
\newtheorem{definition}{Definition}

\allowdisplaybreaks
\usepackage{titling}
\usepackage{array}
\preauthor{\begin{center}
    \large \lineskip .75em%
        \begin{tabular}[t]{>{\centering\arraybackslash}p{.45\textwidth}}}
        \postauthor{\end{tabular}\par\end{center}}
\makeatletter
\renewcommand\and{
  \end{tabular}%
  \hfill
  \begin{tabular}[t]{>{\centering\arraybackslash}p{.45\textwidth}}}
\makeatother

\begin{document}

\title{A new formulation of metriplectic dynamics with an application to quasigeostrophic ocean modeling with advected quantities}

\author{F.J.\ Beron-Vera\thanks{Corresponding author.}\\ Department of Atmospheric Sciences\\ Rosenstiel School of Marine, Atmospheric \& Earth Science\\ University of Miami\\ Miami, Florida, USA\\ \href{mailto:fberon@miami.edu}{fberon@miami.edu} \and E.\ Luesink\\ Korteweg-De Vries Institute\\ University of Amsterdam\\ Amsterdam, The Netherlands\\ \href{mailto:e.luesink@uva.nl}{e.luesink@uva.nl}}

\date{Started: September 2, 2025.  This version: \today.}

\maketitle

\begin{abstract}
    A general formulation of metriplectic dynamics is presented, where the metriplectic four-bracket is constructed by multiplying two skew-symmetric brackets. The new formulation is then used to introduce irreversibility in a generalized two-dimensional (2D) quasigeostrophic (QG) upper-ocean model involving advected quantities, with the thermal QG model as a special case. By construction, the resulting dynamics ensure the conservation of internal energy and the generation of entropy, in accordance with the first and second laws of thermodynamics. Our metriplectic dynamics formulation allows for a flexible specification of irreversibility, ranging from a type that results in nearly material conservation of potential vorticity to the representation of realistic forcing and dissipation in 2D QG ocean modeling with advected quantities.
\end{abstract}

\tableofcontents

\section{Introduction}

In seminal works by \cite{Willems-72}, \cite{Kaufman-84}, \cite{Morrison-86}, \cite{Grmela-Ottinger-97} and \cite{Ottinger-Grmela-97}, frameworks were introduced to describe dynamical systems exhibiting \emph{irreversible} dynamics.
\begin{definition}\label{def:reversible}
    A dynamical system is characterized as \textbf{irreversible} if its dynamics adhere to the first and second laws of thermodynamics, meaning that it conserves internal energy while producing entropy. On the other hand, a system is classified as \textbf{reversible} if both internal energy and entropy are preserved.
\end{definition}
The framework proposed in \cite{Kaufman-84, Morrison-86} combines a generalized Hamiltonian system, which yields conservative, and hence \emph{reversible}, dynamics and a gradient dynamical system in a metric space, which produces dissipative dynamics, to form what was termed a \emph{metriplectic} system. The phase space of a metriplectic system is equipped with a bracket composed of the sum of a Poisson bracket, which is skew-symmetric and satisfies the Jacobi identity, and a symmetric bracket. The symmetric bracket was proposed to be constructed by choosing the metric tensor in such a way to ensure irreversibility when the energy is taken to be the Hamiltonian and the entropy as the Casimir of the Poisson bracket.

In \cite{Morrison-Updike-24}, a framework for introducing irreversibility into an otherwise reversible system was presented by introducing a \emph{metriplectic four-bracket}. Like the Poisson bracket defined on phase space functions, the four-bracket is characterized by its four slots with symmetries and properties motivated by Riemannian curvature. The Hamiltonian ($\mathcal H$) of the reversible system and the Casimir ($\mathcal C$) of the corresponding Poisson bracket $\{\,,\hspace{.05cm}\}$ are then used to induce a symmetric, positive-semidefinite two-bracket $(f,g)_{\mathcal H}$ for any functionals $f,g$. This formulation ensures that irreversible dynamics with respect to internal energy $\mathcal H$ and entropy $\mathcal C$ are controlled by $\dot f = \{f,\mathcal H\} + (\mathcal C,f)_{\mathcal H}$.

Another important framework for building irreversibility is GENERIC (General Equation for Non-Equilibrium Reversible–Irreversible Coupling) \cite{Grmela-Ottinger-97, Ottinger-Grmela-97}. As originally conceived, GENERIC does not fit directly into the metriplectic bracket framework of \cite{Morrison-Updike-24}. However, in \cite{Morrison-Updike-24}, a procedure is given for turning GENERIC brackets into equivalents of metriplectic brackets.

In an important departure from \cite{Morrison-Updike-24}, \cite{Luesink-25} constructs irreversible dynamics based on a deformation of the underlying Poisson bracket without requiring a metric structure, as in \cite{Morrison-Updike-24} or \cite{Bloch-eal-24}. 

The aim of this paper is twofold. The first objective is to introduce a new general formulation of metriplectic dynamics. This formulation builds on the work of \cite{Luesink-25} and uses the product of two skew-symmetric brackets to construct the metriplectic four-bracket, thus it is characterized by its simplicity. The second objective is to use this formulation to introduce irreversibility in a two-dimensional (2D) model that describes low-frequency, i.e., quasigeostrophic (QG), variability in the upper layer of the ocean, including an arbitrary number of advected quantities. The 2D QG model with advected quantities is developed based on a proposal made in \cite{Beron-Luesink-25} to generalize a 2D QG model with lateral buoyancy (temperature) variations and stratification, as developed in \cite{Beron-21-POFb}. This form of ocean modeling, referred to as thermal ocean modeling \cite{Warneford-Dellar-13}, which was popular throughout the 1980s and 1990s \cite{Schopf-Cane-83, McCreary-etal-93, Fukamachi-etal-95, Ripa-JFM-95}, has experienced a significant resurgence in recent years \cite{Holm-etal-21, Lahaye-etal-24, Beron-24-POFa} due to its ability to conceptually explain observed behavior.

The remainder of the paper is organized as follows. Section~\ref{sec:formulation} presents a general formulation of irreversible, thermodynamically consistent dynamics, covering both the case in which the Hamiltonian is not invariant under a symmetry transformation (Section~\ref{sec:nonsym}) and the case in which it is invariant (Section~\ref{sec:sym}). The formulation assumes that the configuration space is an arbitrary finite-dimensional manifold, but the results of Section~\ref{sec:sym} are restricted to two space dimensions, a simplification that serves our aim of specializing the formulation to advQG dynamics. This specialization is developed in Section~\ref{sec:il-reversible}, where key properties of the model are reviewed. In Section~\ref{sec:il-irreversible}, we construct irreversible dynamics for this model based on the proposed formulation. Section~\ref{sec:test} provides a numerical test of the resulting metriplectic dynamics. Conclusions and outlook are presented in Section~\ref{sec:con}. Finally, Appendix~\ref{app:qg} discusses an extension to the standard 2D QG model, which, including no advected quantities, produces dynamics that reside on an invariant subspace of advQG dynamics.

\section{A general formulation of irreversible dynamics}\label{sec:formulation}

\subsection{Setup}

Let $M$ be a finite-dimensional smooth, connected, Riemannian manifold coordinatized by $\mathbf x$ with volume form $\mu(d\mathbf x)$; let $t \in \mathbb{R}$ represent time; and let $u(\mathbf x,t)$ symbolize an $n$-tuple, $n \ge 1$ finite, of smooth time-dependent real functions (scalar fields) on the \emph{configuration space} $\Omega \subseteq M$. The vector space formed by such functions is denoted as $C^\infty(\Omega)$, and we reserve the notation $\mathfrak{F}_\Omega^n$ for the infinite-dimensional \emph{phase space}, which is composed of $n$ Cartesian copies of $C^\infty(\Omega)$. Thus, $u_i \in C^\infty(\Omega)$, while $u \in \mathfrak{F}_\Omega^n = C^\infty(\Omega)^n$. We leave the topology on $\mathfrak F_\Omega^n$ unspecified, except that we equip $\mathfrak F_\Omega^n$ with an $L^2$ inner product $\langle \, , \hspace{.01cm} \rangle : \mathfrak F_\Omega^n \times \mathfrak F_\Omega^n \to \mathbb{R}$, defined in the usual way as
\begin{equation}
    \langle u, v\rangle := \sum_{i,j=1}^n\int_\Omega g_{ij}u_iv^j\,\mu(d\mathbf x) 
\end{equation}
for all $u,v \in \mathfrak F_\Omega^n$. Here we have used the Riemannian metric inside the integral. If the metric tensor is diagonal, the double sum collapses to a single sum. If $\Omega = M$ for $M$ noncompact, then sufficient fall-off conditions are assumed for all $u \in \mathfrak{F}_\Omega^n$, ensuring that boundary terms arising from integration by parts over $\Omega$ vanish identically. In case $\Omega=M$ for compact $M$ without boundary, we need not put boundary conditions. If $\Omega \subset M$, we will assume that it is bounded by $\partial\Omega$. In such a case, $u$ must satisfy appropriate conditions on $\partial\Omega$ for such an operation to be meaningful. A functional $f$ of $u$, that is, a smooth map $u \in \mathfrak F_\Omega^n \to \mathbb R$, will be denoted $f[u]$.  The vector space of functionals on $\mathfrak F_\Omega^n$ will therefore be denoted $C^\infty(\mathfrak F_\Omega^n)$. Here, we consider functionals of $u$ (henceforth implicitly understood when not stated explicitly) of the form
\begin{equation}
    f[u] = \int_\Omega F(u,\nabla u,\nabla\nabla u,\dotsc,\mathbf x)\,\mu(d\mathbf x)
\end{equation}
for some (smooth) function $F$. The first variation of $f$ at $u$ in the arbitrary direction $\delta u$, called a variation of $u$, is defined as the linear map
\begin{equation}
    \delta f[u] \cdot \delta u := \left.\frac{d}{d\varepsilon}\right\vert_{\varepsilon=0}f[ u+\varepsilon \delta u] = \left\langle\frac{\delta f}{\delta u}, \delta u\right\rangle.
    \label{eq:dfdu}
\end{equation}
Here, $\frac{\delta f}{\delta u} \in \mathfrak F_\Omega^n$ provides a unique representation for $\delta f[u] \in C^\infty(\mathfrak F_\Omega^n)$, as guaranteed by the Riesz theorem, and is known as the functional derivative of $f$ at $u$.  

Consider the \emph{Poisson bracket} $\{\,,\hspace{.01cm}\} : C^\infty(\mathfrak F_\Omega^n) \times  C^\infty(\mathfrak F_\Omega^n) \to  C^\infty(\mathfrak F_\Omega^n)$ \cite[e.g.,][Section IV.A]{Morrison-98}
\begin{equation}
    \{f,g\} : = \left\langle\frac{\delta f}{\delta u}, \mathbf J \frac{\delta g}{\delta u}\right\rangle,
\end{equation}
where $\mathbf{J}(u)(\hspace{.15cm})$ is a skew-adjoint bivector operator, referred to as a \emph{Poisson operator}. More precisely, the Poisson bracket satisfies, for any functionals $f,g,h$ and $\lambda \in \mathbb R$, the following properties: 
\begin{enumerate}
    \item[\textbf{P1.}] linearity in all arguments, e.g., $\{f+\lambda g,h\} = \{f,h\} + \lambda \{g,h\}$; 
    \item[\textbf{P2.}] skew-symmetry, viz., $\{f,g\} = - \{g,f\}$; 
    \item[\textbf{P3.}] Jacobi identity, viz., $\{\!\{f,g\},h\} + \{\!\{h,f\},g\} + \{\!\{g,h\},f\} = 0$; and 
    \item[\textbf{P4.}] Leibniz rule in each argument, e.g., $\{fg,h\} = \{f,h\}g + f\{g,h\}$. 
\end{enumerate}
Property P1 follows from the linearity of $\mathbf J$, whereas property P4 follows from this and the property of the (functional) derivative. Property P2 is implied by the skew-adjointness of $\mathbf J$, while P3 imposes a restriction on $\mathbf J$ beyond skew-adjointness. The pair $(\mathfrak F_\Omega^n, \{\,,\hspace{.01cm}\})$ forms a \emph{Poisson manifold}. The Jacobi identity turns $(C^{\infty}(\mathfrak F_\Omega^n), \{\,,\hspace{.01cm}\})$ into a \emph{Lie algebra}, while the Leibniz rule further extends it into a \emph{Poisson algebra}. The evolution of an arbitrary functional $f$ under \emph{generalized Hamiltonian dynamics} is controlled by
\begin{equation}
    \dot f = \{f,\mathcal H\} \Longleftrightarrow \partial_t u = \mathbf J\frac{\delta \mathcal H}{\delta u},
    \label{eq:sys}
\end{equation}
where $\mathcal H$ is the \emph{Hamiltonian}. By the skew-symmetry of the Poisson bracket, the Hamiltonian is an integral of the motion:
\begin{equation}
    \dot{\mathcal H} = \{\mathcal H,\mathcal H\} = 0.
\end{equation}

Let $\mathcal M$ satisfy, for $\varepsilon > 0$ small,
\begin{equation}
    \varepsilon \{f,\mathcal M\} = \delta_{\mathcal M}f
    \label{eq:inf}
\end{equation}
for any functional $f$, where $\delta_{\mathcal M}f$ is the infinitesimal transformation induced by a (Lie) group action on space-time domain $\Omega \times \mathbb R$.  In \cite{Shepherd-90}, the functional $\mathcal M$ is referred to as the \emph{generator of an infinitesimal transformation}. This shares a lineage with the notion of momentum map, which plays a central role in the study of finite-dimensional Hamiltonian systems with continuous symmetry in connection with Noether's theorem \cite[cf., e.g.,][Chapter 11]{Marsden-Ratiu-99}; we return to this below. Evidently, $\mathcal M$ is conserved under the Hamiltonian dynamics \eqref{eq:sys} if $\delta_{\mathcal M}\mathcal H = 0$, that is, the Hamiltonian is symmetric with respect to the group action, which represents \emph{Noether's theorem on the Hamiltonian side}.  Now, differentiate \eqref{eq:inf} with respect to time and subtract $\delta_{\mathcal M}\dot f$ from the result to find:
\begin{equation}
    \tfrac{d}{dt}\delta_{\mathcal M}f - \delta_{\mathcal M}\dot f = \varepsilon\{f,\dot{\mathcal M}\}.
\end{equation}
If $\mathcal M$ is conserved, then applying a transformation and allowing time evolution are independent of the order in which these operations are performed, indicating a quite general symmetry. The reverse does not necessarily hold: if these operations commute, then $\mathcal M$ is equal to a function of a functional $\mathcal C$, commonly referred to as a \emph{Casimir}, forming the kernel of the Poisson bracket. Namely, it satisfies
\begin{equation}
    \{f,\mathcal C\} = 0\, \forall f \Longleftrightarrow \mathbf J\frac{\delta \mathcal C}{\delta u} = 0.
\end{equation}
Clearly, Casimirs are conserved under the Hamiltonian dynamics:
\begin{equation}
    \dot{\mathcal C} = \{\mathcal C,\mathcal H\} = 0.
\end{equation}
This may be regarded as a broader formulation of Noether's theorem, as discussed in \cite{Ripa-RMF-92a, Ripa-JFM-92b}.   

Now, consider an additional bilinear bracket $\{\hspace{-.115cm}\{\,,\hspace{.01cm}\}\hspace{-.115cm}\} : C^\infty(\mathfrak F_\Omega^n) \times C^\infty(\mathfrak F_\Omega^n) \to C^\infty(\mathfrak F_\Omega^n)$ that satisfies the Leibniz rule, but is only required to be skew-symmetric:
\begin{equation}
    \{\hspace{-.115cm}\{f,g\}\hspace{-.115cm}\} := \left\langle\frac{\delta f}{\delta u}, \mathbb J\frac{\delta g}{\delta u}\right\rangle = -\{\hspace{-.115cm}\{g,f\}\hspace{-.115cm}\}
    \label{eq:bra2}
\end{equation}
for any functionals $f,g$ and some (skew-adjoint bivector) operator $\mathbb J$. Using this bracket, define a four-bracket $(\,,\,;\,,\,) : C^\infty(\mathfrak F_\Omega^n) \times C^\infty(\mathfrak F_\Omega^n) \times C^\infty(\mathfrak F_\Omega^n) \times C^\infty(\mathfrak F_\Omega^n) \to C^\infty(\mathfrak F_\Omega^n)$ by:
\begin{equation}
    (f,h;g,k) := \{\hspace{-.115cm}\{f,h\}\hspace{-.115cm}\}\{\hspace{-.115cm}\{g,k\}\hspace{-.115cm}\}
    \label{eq:four}
\end{equation}
for all functionals $f,g,h,k$.  This bracket is quadrilinear and satisfies the Leibniz rule, both properties being implied by $\{\hspace{-.115cm}\{\,,\hspace{.01cm}\}\hspace{-.115cm}\}$. Define a last bilinear bracket $(\,,\hspace{.01cm})_h : C^\infty(\mathfrak F_\Omega^n) \times C^\infty(\mathfrak F_\Omega^n) \to C^\infty(\mathfrak F_\Omega^n)$ using the four-bracket \eqref{eq:four}, thereby representing a derivation, by setting its third and fourth arguments equal, namely,
\begin{equation}
    (f,g)_h := (f,h;g,h).
    \label{eq:two}
\end{equation}
By construction, this bracket satisfies, for any functionals $f,g,h$, the following properties:
\begin{enumerate}
    \item[\textbf{M1.}] $(f,g)_h = (g,f)_h$ (symmetry);
    \item[\textbf{M2.}] $(f,h)_h = 0$; and
    \item[\textbf{M3.}] $(f,f)_h \ge 0$ (positive-semidefiniteness).
\end{enumerate}
These represent the axioms of \emph{metriplecticity} \cite{Morrison-Updike-24}.  When properties M1--M3 hold, we refer to \eqref{eq:four} as the \emph{metriplectic four-bracket} and \eqref{eq:two} as the \emph{metriplectic induced two-bracket}, respectively.

In conclusion, we acknowledge that the metriplectic four-bracket \eqref{eq:four} is inspired by \cite{Luesink-25}. This work presents a geometric construction of irreversible dynamics on Poisson manifolds, where a metriplectic four-bracket is obtained by multiplying two Poisson 2-cocycles, derived by deforming the Poisson algebra via central extension. Our mathematical construction is similar, in that we introduce two antisymmetric objects and take their product. It is not necessary that these objects themselves are Poisson brackets. We note that the approach of \cite{Luesink-25} differs from that of \cite{Morrison-Updike-24, Zaidni-Morrison-25}. The authors construct four-brackets from the Kulkarni--Nomizu (KN) product of two symmetric bilinear brackets. Introduced in Riemannian geometry \cite{Lee-18}, the KN product of two symmetric tensors possesses the same algebraic properties as the Riemann tensor. These properties are analogous to the metriplecticity properties M1--M2, which inspired its application.

\subsection{Nonsymmetric Hamiltonian systems}\label{sec:nonsym}

Assume that the Hamiltonian is reduced, that is, all its symmetries have been removed by changing to appropriate variables using the momentum maps associated with the symmetries. Of course, it remains symmetric under time shifts, since the Hamiltonian is itself the conserved infinitesimal generator of such symmetry shifts.

\begin{proposition}\label{pro:nonsym}
    Let the internal energy, $\mathcal U$, be taken to be the Hamiltonian and the entropy, $\mathcal S$, be given by a Casimir. Irreversible dynamics are controlled by 
    \begin{subequations}\label{eq:sys-irr}
    \begin{equation}
        \dot f = \{f,\mathcal H\} + (\mathcal C,f)_{\mathcal H}
    \end{equation}
    for any functional $f$ and in particular for $f =  u$,
    \begin{equation}
       \partial_t u = \big(\mathbf J + \{\hspace{-.115cm}\{\mathcal C,\mathcal H\}\hspace{-.115cm}\}\mathbb J\big)\frac{\delta \mathcal H}{\delta u}.
    \end{equation}
    \end{subequations}
\end{proposition}

\begin{proof}
    Let $\mathcal U = \mathcal H$. Then, $\{\mathcal U,\mathcal H\} = 0$ due to the skew-symmetry of the Poisson bracket, and $(\mathcal C,\mathcal U)_{\mathcal H} = 0$ due to the metriplecticity of $(\,,\hspace{.01cm})_{\mathcal H}$. Consequently, $\mathcal U$ is conserved, regardless of the nature of $\mathcal S$. Now, let $\mathcal S = \mathcal C$. Then, we have $\{\mathcal S,\mathcal H\} = 0$ because $\mathcal C$ Poisson commutes, and $(\mathcal S,\mathcal S)_{\mathcal H} \ge 0$ by metriplecticity. Consequently, $\mathcal S$ is produced or preserved, but never consumed.
\end{proof} 

\begin{remark}\label{rem:nonsym}
    A few observations are in order. 
    \begin{enumerate}
        \item Set
        \begin{equation}
            f = \mathcal H - T\mathcal C \equiv \mathcal U - T\mathcal S =: \mathcal A, 
        \end{equation}
        where $T$ is a constant. The functional $\mathcal A$ can be interpreted as a \textup{(}Helmholtz\textup{)} free energy with $T$ playing the role of absolute temperature, as discussed in \textup{\cite{Morrison-86}}.  In essence, free energy represents the availability of energy in a thermodynamic system to perform useful work at a constant temperature and volume. While internal energy is conserved, spontaneous processes reduce this availability as entropy increases, so that the free energy decreases and reaches a minimum at equilibrium. Consistent with this, from \textup{(}\ref{eq:sys-irr}a\textup{)} it follows that $\dot{\mathcal A} \le 0$. 
        \item As noted in \textup{\cite{Morrison-86}}, conditional equilibria of the reversible dynamics produced by the Hamiltonian system \eqref{eq:sys}, namely, $u = U$ such that
        \begin{equation}
            \left.\frac{\delta \mathcal A}{\delta u}\right\vert_{u=U} = 0 \Longrightarrow \partial_t U = 0,
        \end{equation}
        represent absolute equilibria of the irreversible dynamics produced by the metriplectic system \eqref{eq:sys-irr}.
        \item The equation $\dot{f} = \{f,\mathcal H\} + (\mathcal C,f)_{\mathcal C}$ evidently generates irreversible dynamics with consistent thermodynamics with respect to $\mathcal U = \mathcal C$ and $\mathcal S = \mathcal H$. However, the interpretation of $\mathcal U$ as a internal energy is lost.
    \end{enumerate} 
\end{remark}

\subsection{Symmetric Hamiltonian systems}\label{sec:sym}

We are interested in the case where the Hamiltonian is invariant under a group action on space domain $\Omega$ by translation. For simplicity, we consider $M = \mathbb{R}^2$, which suits our purposes well, as the type of irreversible ocean dynamics with thermodynamic consistency we seek to construct is defined in two dimensions.  Call $\chi$ the direction on which such a symmetry group acts.  For instance, if $\Omega$ has axial symmetry, then such a direction is the azimuthal direction in polar coordinates, and a translation in that direction represents a rotation. Denote by $\mathcal{M}^\chi$ the \emph{$\chi$-momentum}, which generates the infinitesimal symmetry transformation $\chi \mapsto \chi + \varepsilon$. Under this transformation, $u \mapsto u + \varepsilon \partial_\chi u$, implying $\delta_{\mathcal{M}^\chi} u = \varepsilon \partial_\chi u$. Then, from \eqref{eq:inf}, it follows that $\mathcal{M}^\chi$ satisfies
\begin{equation}
    \partial_\chi u = \mathbf J\frac{\delta \mathcal M^\chi}{\delta u},
    \label{eq:M}
\end{equation}
and, by assumption, is preserved under the Hamiltonian dynamics. (The \emph{$\chi$-momentum} is sometimes \cite{Shepherd-90} defined as minus $-\mathcal M^\chi$.) Consequently, 
\begin{equation}
    \mathcal H_{\dot \chi} : = \mathcal H + \dot \chi \mathcal M^\chi, \quad \dot \chi = \text{const},
    \label{eq:pseudo}
\end{equation}
is a Hamiltonian for the dynamics as observed from a reference frame that steadily translates or rotates in the $\chi$-direction at linear or angular speed $\dot \chi$.  Specifically,
\begin{equation}
    \partial_t u + \dot \chi \partial_\chi u = \mathbf J\frac{\delta \mathcal H_{\dot \chi}}{\delta u}
\end{equation}
of, for any functional $f$,
\begin{equation}
    \dot f = \{f,\mathcal H_{\dot \chi}\}.
\end{equation}
We refer to $\mathcal H_{\dot \chi}$ in \eqref{eq:pseudo} as an \emph{$\chi$-symmetry induced Hamiltonian}. 

\begin{proposition}\label{pro:sym}
    Let $\mathcal U = \mathcal H_{\dot \chi}$ and $\mathcal S = \mathcal C$. As seen from a frame moving in the $\chi$-direction at constant speed $\dot s$, irreversible dynamics are controlled by 
    \begin{equation}
        \dot f = \{f,\mathcal H_{\dot \chi}\} + (\mathcal C,f)_{\mathcal H_{\dot \chi}}
    \end{equation}
    for any functional $f$ and in particular for $f =  u$,
    \begin{equation}
       \partial_t u + \dot \chi \partial_\chi u = \big(\mathbf J + \{\hspace{-.115cm}\{\mathcal C,\mathcal H_{\dot \chi}\}\hspace{-.115cm}\}\mathbb J\big)\frac{\delta \mathcal H_{\dot \chi}}{\delta u}.
\end{equation}
\end{proposition}

\begin{proof}
    The proof mirrors that of Proposition \ref{pro:nonsym}, substituting $\mathcal H$ with $\mathcal H_{\dot \chi}$.
\end{proof} 

Finally, we note that Remark \ref{rem:nonsym}.2 remains the same, and Remark \ref{rem:nonsym}.3 holds true, mutatis mutandis, with $\mathcal{H}$ replaced by $\mathcal{H}_{\dot{\chi}}$, while Remark \ref{rem:nonsym}.1 must be reinterpreted with this replacement.

\section{Reversible advQG dynamics}\label{sec:il-reversible}

In this section, we develop the reversible form of what we have termed \emph{2D QG dynamics with advected quantities} or \emph{advQG dynamics}. This is done by building on a proposition from our recent work \cite{Beron-Luesink-25} to generalize a 2D QG model with lateral density inhomogeneity and stratification \cite{Beron-21-POFb}. The model of \cite{Beron-21-POFb} was derived to improve the description of (low-frequency) variability in the upper layer ocean. For future reference, the model of \cite{Beron-21-POFb} was called IL$^{(0,\alpha)}$QG to reflect that it consists of an Inhomogeneous Layer of fluid where, as in the standard 2D QG model \cite{Pedlosky-87}, referred to as HLQG for having a Homogeneous Layer of fluid, the velocity does not vary with the vertical, indicated by the $0$ in the superscript. Unlike the the HLQG, the density varies as a polynomial of order $\alpha \ge 0$ in the vertical, with coefficients that are arbitrary functions of horizontal position and time, as they are advected under the flow.

\subsection{Formulation}

Let $\mathbf{x} = (x,y)$ denote Cartesian position on a fluid domain $\Omega$ in $M = \mathbb{R}^2$, with volume form $\mu(d\mathbf x) = dxdy$, representing the infinite $f$- or $\beta$-plane. The former is characterized by a constant Coriolis parameter $f = f_0$ (twice the local Earth's angular velocity), while the latter is characterized by a Coriolis parameter that varies linearly with latitude ($y$), viz., $f = f_0 + \beta y$. 

Let $\psi(\mathbf x,t) \in C^\infty(\Omega)$ denote \emph{streamfunction} and $a(\mathbf x,t) \in C^\infty(\Omega)^{\alpha+1}$ represent an $(\alpha+1)$-tuple, $\alpha \ge 0$, of \emph{advected quantities}, that is, they satisfy
\begin{equation}
    \partial_t a + J(\psi,a) = 0,
\end{equation}
where 
\begin{equation}
    J(v,w) := \nabla^\perp v\cdot\nabla w  = -\partial_yv\partial_xw + \partial_xv\partial_yw
\end{equation}
for any $v,w \in C^\infty(\Omega)$, which is the Jacobian of the map $\mathbf x \mapsto (v(\mathbf x,t),w(\mathbf x,t))$.  

Let $q(\mathbf x,t) \in C^\infty(\Omega)$ represent the \emph{potential vorticity}, in a generalized sense described by
\begin{equation}
    q = \nabla^2\psi - R^{-2}\psi + A(a) + f,
    \label{eq:q}
\end{equation}
The function $A$ is arbitrary.  The constant $R>0$ is given the interpretation of equivalent barotropic Rossby deformation radius, assuming a \emph{reduced-gravity }setting.  Relevant to describe surface-intensified variability in the ocean, in such a setting the active fluid layer floats atop a quiescent, infinitely deep layer.  

\begin{assumption}\label{ass:Omega}
If $\Omega = \mathbb{R}^2$, we assume $f = f_0$ and require $\psi$ and $a$ to decay sufficiently for integration by parts to be applicable. When $\Omega \subset \mathbb{R}^2$, we assume $\partial \Omega$ represents a solid coast. To ensure complete self-consistency, we consider three possibilities.
\begin{enumerate}
    \item The domain $\Omega$ represents a zonal channel specifically oriented in the $x$ (longitude) direction. It can be either infinitely long, $\Omega \cong \mathbb R \times [0,1]$, or periodic, $\Omega \cong \mathbb R/\mathbb Z \times [0,1]$.  The domain boundary $\partial \Omega$ is the union of the channel's coastlines.
    \item The domain $\Omega$ is bounded and exhibits axial symmetry or symmetry in the azimuthal direction $\left( \vartheta = \tan^{-1}\frac{y}{x} \right)$. In this case, $\Omega \cong S^1$ and we set $f = f_0$.
    \item The domain $\Omega$ is enclosed by an arbitrarily shaped coastline $\partial \Omega_0$ and includes possibly multiple irregular islands, say $I$.  The boundary of the latter, multiply connected domain is $\partial \Omega = \bigcup_{i=0}^I \partial \Omega_i$. In this context, we also restrict to the $f$-plane.
\end{enumerate}
In any scenario, the following conditions hold along $\partial \Omega$ (if nonempty):
\begin{equation}
    \partial_s\psi = 0,\quad
    \partial_sa = 0,
    \label{eq:bc}
\end{equation}
where $\partial_s = \nabla\cdot\hat{\mathbf{n}}^\perp\vert_{\partial \Omega}$, and $\hat{\mathbf{n}}\vert_{\partial \Omega}$ is the outward unit normal to $\partial \Omega$ \textup{(}or to each connected piece\textup{)}.
\end{assumption}

Boundary conditions \eqref{eq:bc} respectively indicate the absence of normal flow through \(\partial\Omega\) and the iso-\(a\) character of \(\partial\Omega\), which are necessary to ensure that integration by parts yields no boundary terms. In other words, these are the natural conditions ensuring that the system has no exchange with the ``environment,'' i.e., that it is a closed system. This allows for the establishment of an Euler--Poincar\'e/Lie--Poisson dual formulation for reversible advQG dynamics, which we discuss below. If we choose to elevate this dual formulation in constructing irreversible advQG dynanics, $\Omega$ may be bounded by an irregular coastline on the $\beta$-plane, including the possibility of the presence of multiple connections, i.e., islands.  In this context, the iso-$a$ nature of $\partial\Omega$, specified by the boundary condition on the right-hand side of \eqref{eq:bc}, is not necessary. An Euler--Poincar\'e/Lie--Poisson dual formulation for reversible advQG dynamics may still be possible, but this would require deviation from the variational calculus implied by \eqref{eq:dfdu}, for example, as discussed in \textup{\cite{Lewis-etal-86}}.

\begin{assumption}\label{ass:gamma}
Suppose that $\Omega \subset \mathbb R^2$ is bounded.  Consider the circulation along $\partial \Omega$, defined by
\begin{equation}
    \gamma := \oint_{\partial \Omega} \nabla^\perp\psi\cdot d\mathbf x.
\end{equation}
We will assume that
\begin{equation}
    \dot\gamma = 0,\quad \delta\gamma = 0.
    \label{eq:gamma}
\end{equation}
\end{assumption}

The condition on the left-hand side of \eqref{eq:gamma} ensures energy conservation when $R < \infty$, meaning the bottom boundary is free. Otherwise, $\dot\gamma = 0$ is inherently satisfied by the dynamics. The condition on the right of \eqref{eq:gamma} enables consistency between the variational calculus and \eqref{eq:dfdu}.  It is implicit that if $\partial \Omega$ is the union of connected pieces as described in Assumptions \ref{ass:Omega}.1 and \ref{ass:Omega}.3, then there will be as many circulations as connected pieces, and conditions \eqref{eq:gamma} will apply to each of them.

Irreversible advQG dynamics are the Hamiltonian dynamics on the phase space variables $u = (q,a) \in C^\infty(\Omega) \times C^\infty(\Omega)^{\alpha+1} = \mathfrak F_\Omega^{\alpha+2}$ produced by the Hamiltonian (energy) given by 
\begin{equation}
    \mathcal H[q,a] := \frac{1}{2}\int_\Omega |\nabla\psi|^2 + R^{-2}\psi^2\,dxdy = \frac{1}{2}\left\langle q,(\nabla^2-R^{-2})^{-1}(f + A(a) - q)\right\rangle
    \label{eq:H}
\end{equation}
and the Poisson bracket with Poisson operator given by
\begin{equation}
    \mathbf J := -
    \begin{bmatrix} 
    J(q,\cdot\hspace{.05cm}) & J(a,\cdot\hspace{.05cm})\\
    J(a,\cdot\hspace{.05cm}) & 0
    \end{bmatrix}.
    \label{eq:J}
\end{equation}
The last equality in \eqref{eq:H} follows from an integration by parts, with $(\nabla^2-R^{-2})^{-1}$ appropriately interpreted using the Green's function of the elliptic problem \eqref{eq:q}. On the other hand, using \eqref{eq:J}, the Poisson bracket takes the explicit form
\begin{subequations}\label{eq:PB}
\begin{align}
    \{f,g\} 
    &= 
    - \left\langle\frac{\delta f}{\delta q}, J\left(q,\frac{\delta g}{\delta q}\right) + J\left(a,\frac{\delta g}{\delta a}\right)\right\rangle
    - \left\langle\frac{\delta f}{\delta a}, J\left(a,\frac{\delta g}{\delta q}\right)\right\rangle
    \label{eq:PB-noLP}\\
    &= \left\langle q, J\left(\frac{\delta f}{\delta q},\frac{\delta g}{\delta q}\right)\right\rangle + \left\langle a, J\left(\frac{\delta f}{\delta q},\frac{\delta g}{\delta a}\right) +J\left(\frac{\delta f}{\delta a},\frac{\delta g}{\delta q}\right) \right\rangle
    \label{eq:PB-LP}
\end{align}
\end{subequations}
for all $f,g \in C^\infty(\mathfrak F_\Omega^{\alpha+2})$, where the last equality follows an integration by parts. Upon computing 
\begin{equation}
    \frac{\delta \mathcal H}{\delta q} = -\psi,\quad
    \frac{\delta \mathcal H}{\delta a} = \psi \partial_aA,
    \label{eq:dHqa}
\end{equation}
the following evolution equations producing advQG dynamics are obtained:
\begin{subequations}\label{eq:advQG}
    \begin{align}
    \partial_tq &= \{q,\mathcal H\} = - J(\psi,q-A),\label{eq:advQG-a}\\
    \partial_ta &= \{a,\mathcal H\} = - J(\psi,a),\label{eq:advQG-b}
\end{align}
\end{subequations}
where the relationship
\begin{equation}
    J(a,\psi\partial_aA) = J(a,\psi)\partial_aA + J(a,\partial_aA)\psi = J(a,\psi)\partial_aA = -J(\psi,A) 
    \label{eq:Ja}
\end{equation}
was used to obtain rightmost equality in \eqref{eq:advQG-a}.

\begin{remark}[On notation]
    In this paper, when $v, w \in C^\infty(\Omega)^n$ with $n \geq 1$, the product $vw$ is understood as $\sum_{i=1}^n v_i w_i$. For instance, $J(a, \psi \partial_a A)$ refers to $\sum_{i=1}^{\alpha+1} J(a_i, \psi \partial_{a_i} A)$.
\end{remark}

\begin{remark}\label{rem:geo}
    This remark pertains to the geometric formulation of advQG dynamics. The formulation adheres to \textup{\cite{Marsden-Morrison-84, Beron-Luesink-25}}, with additional details and a correction to \textup{\cite{Beron-Luesink-25}} provided here. In this formulation, the phase space $\mathfrak F_\Omega^{\alpha+2}$ is viewed as $\mathfrak{sdiff}(\Omega)^* \ltimes C^\infty(\Omega)^{\alpha+1}$ with $q \in \mathfrak{sdiff}(\Omega)^*$ and $a \in C^\infty(\Omega)^{\alpha+1}$, by identifying the dual of $C^\infty(\Omega)$ with itself using the $L^2$ inner product on $\Omega$. Here, $\mathfrak{sdiff}(\Omega) \ltimes C^\infty(\Omega)^{\alpha+1}$ represents the \textbf{Lie algebra of the Lie group} $\mathrm{SDiff}(\Omega) \ltimes C^\infty(\Omega)^{\alpha+1}$, which is obtained by extending the Lie group of symplectic \textup{(}i.e., area preserving\textup{)} diffeomorphisms on $\Omega$, $\mathrm{SDiff}(\Omega)$, via \textbf{semidirect product} with $C^\infty(\Omega)^{\alpha+1}$. \textup{(}For an in-depth survey of how broken symmetry leads to semidirect product Lie algebras, readers are referred to \textup{\cite{Holm-24}}.\textup{)} 
    
    As a set \textup{(}manifold\textup{)}, $\mathrm{SDiff}(\Omega) \ltimes C^\infty(\Omega)^{\alpha+1}$ is $\mathrm{SDiff}(\Omega) \times C^\infty(\Omega)^{\alpha+1}$, with multiplication given by the $(\alpha+2)$-tuple
    \begin{equation}
        (g_t,a)(\bar g_t,\bar a) = (g_t\circ \bar g_t, a + \bar a\circ g_t^{-1})
    \end{equation}
    for all $g_t,\bar g_t \in \mathrm{SDiff}(\Omega)$ and $a,\bar a \in C^\infty(\Omega)^{\alpha+1}$, where the representation of $\mathrm{SDiff}(\Omega)$ on $C^\infty(\Omega)^{\alpha+1}$ is given by right action, specifically pushforward. The inverse of $\mathrm{SDiff}(\Omega) \ltimes C^\infty(\Omega)^{\alpha+1}$ is given by $(g_t,a)^{-1} = (g_t^{-1},-ag_t)$, implying that the identity is $(g_0,0) =: e$.  Here, $g_t \in \mathrm{SDiff}(\Omega)$ is the $(0, t)$-flow map, which takes particle positions from their initial positions at time $t = 0$ to their current positions at time $t$. Consequently, the velocity is the Hamiltonian vector field $X_\psi = \nabla^\perp\psi\cdot\nabla$ computed from $g_t$ as $X_\psi = \dot{g}_t \circ g_t^{-1}$.  Since $X_\psi$ is fully specified by $\psi$, one regards $\psi$ as an element of $\mathfrak{sdiff}(\Omega) \cong T_{g_0}\mathrm{SDiff}(\Omega)$. To ensure the analytical validity of the geometric formulation in the sense of \textup{\cite{Ebin-Marsden-70}}, $\mathrm{SDiff}(\Omega)$ must be taken to be of Sobolev class $H^k$, $k > 2$, or of H\"older class $C^{k+\alpha}$ with $k\geq 1$ and $0<\alpha<1$. The Hilbert space nature of the Sobolev class diffeomorphisms is convenient, so we restrict to this setting.
    
    Let
    \begin{equation}
        c_t := (g_t,ta),
    \end{equation}
    satisfying $c_0 = (g_0,0) \equiv e$ and $\left.\frac{d}{dt}\right\vert_{t=0}c_t = (X_\psi, a)$. The Lie bracket in $\mathfrak{sdiff}(\Omega) \ltimes C^\infty(\Omega)^{\alpha+1}$ is the $(\alpha+2)$-tuple defined by
    \begin{subequations}
    \begin{align}
        \left.\frac{d}{dt}\frac{d}{ds}\right\vert_{t,s=0} c^{-1}_tc_sc_t &\doteq \big((J(\psi,\bar\psi),  J(\psi,\bar a_1) + J(a_1,\bar\psi), \dotsc, J(\psi,\bar a_{\alpha+1}) + J(a_{\alpha+1},\bar\psi)\big)\label{eq:lie-a}\\ &=: [(\psi,a),[(\bar\psi,\bar a)]\label{eq:lie-b}
    \end{align}
    \end{subequations}
    for all $\psi,\bar\psi \in \mathfrak{sdiff}(\Omega)$ and $a,\bar a \in C^\infty(\Omega)^{\alpha+1}$. The notation $\doteq$ is used instead of $=$ because the differentiation of the conjugation in \eqref{eq:lie-a} gives $X_{J(\psi,\bar\psi)}$ in the first entry of the tuple, which is determined by $J(\psi,\bar\psi)$. A Lie bracket is presented in \textup{\cite{Beron-Luesink-25}} that differs from \eqref{eq:lie-b}. Each $-$ sign in the expression given in \textup{(}71\textup{)} of \textup{\cite{Beron-Luesink-25}} should be a $+$ sign, which we take the opportunity to rectify. 
    
    The Lie bracket \eqref{eq:lie-b} carries in its dual the \textbf{Lie--Poisson bracket}
    \begin{equation}
        \{f,g\} = \left\langle u, \left[\frac{\delta f}{\delta u}, \frac{\delta g}{\delta u}\right]\right\rangle,
        \label{eq:LP}
    \end{equation}
    where the $L^2$ inner product on $\Omega$ provides a duality paring on $\mathfrak{sdiff}(\Omega) \ltimes C^\infty(\Omega)^{\alpha+1}$. The Lie--Poisson bracket \eqref{eq:LP} is equivalent to \eqref{eq:PB} when specialized to $u = (q,a)$. 
    
    We close by noting that this geometric formulation of the Lie--Poisson bracket is special because the elements of $\mathfrak{sdiff}(\Omega)$ are actually time-dependent smooth functions on $\Omega$, meaning they belong to $C^\infty(\Omega)$. Consequently, the dual of $\mathfrak{sdiff}(\Omega)$ is identified with itself using the $L^2$ inner product on $\Omega$.
\end{remark}

\subsection{Special types of advQG dynamics}\label{sec:ILQG}

The IL$^{(0,\alpha)}$QG \cite{Beron-21-POFb}is a special case of advQG dynamics, as noted above. Two relevant cases, with the interpretation of its variables presented in Table \ref{tab:special}, are discussed here for future reference.  These are the IL$^0$QG, also known as Ripa's model \cite{Dellar-03}, thermal QG \cite{Warneford-Dellar-13} or TQG \cite{Holm-etal-21}, and the IL$^{(0,1)}$QG, which is a stratified form of IL$^0$QG.

\begin{table}[t!]
    \centering
    \begin{tabular}{lcccc}
        \hline\hline
        $\alpha$ &  Model & $A(a)$ & $h - H$ & $\vartheta - g'$\\
        $0$ & IL$^0$QG & $\frac{1}{R^2}\psi_\sigma$ & $\frac{H}{f_0R^2}(\psi-\psi_\sigma)$ & $\frac{2g'}{f_0R^2}\psi_\sigma$\\
        $1$ & IL$^{(0,1)}$QG & $\frac{1}{R^2}\big(\psi_\sigma - \frac{2}{3}\psi_{\sigma^2}\big)$ & $\frac{H}{f_0R^2}(\psi-\psi_\sigma+\frac{2}{3}\psi_{\sigma^2})$ & $\frac{2g'}{f_0R^2}\psi_\sigma +  \sigma(\frac{1}{2}N^2H + \frac{4g'}{f_0R^2}\psi_{\sigma^2})$ \\
        \hline
    \end{tabular}
    \caption{Two relevant special cases of advQG dynamics, the IL$^0$QG and IL$^{(0,1)}$QG models.}
    \label{tab:special}
\end{table}

The IL$^0$QG, which corresponds to advQG dynamics with $\alpha = 0$, has $a_1 = \psi_\sigma$ and $A(a_1) = R^{-2}\psi_\sigma$. This advected quantity represents a rescaled buoyancy ($\vartheta$) deviation from a background value, $g'$. The buoyancy is defined as the negative ratio of the density difference between the active and inert layers to the reference density, used in the Boussinesq approximation, multiplied by gravity.  When density changes are dominated by temperature changes, $\vartheta$ can be interpreted as temperature, clarifying the nomenclature ``thermal QG'' in \cite{Warneford-Dellar-13}. The deviation of the layer thickness ($h$) from the thickness of the fluid at rest ($H$) is proportional to $\psi - \psi_\sigma$. Finally, by the thermal-wind balance, which dominates at low frequencies, it turns out that the (horizontal) velocity in the IL$^0$QG has an \emph{implicit} vertical shear proportional to $\nabla^\perp\psi_\sigma$. 

The IL$^{(0,1)}$QG, a special case of advQG dynamics with $\alpha = 1$, has, in addition to $a_1 = \psi_\sigma$, an another advected quantity, $a_2 = \psi_{\sigma^2}$.  In this model, $A(a_1,a_2) = R^{-2}\big(\psi_\sigma - \frac{2}{3}\psi_{\sigma^2}\big)$, where $R$ is slightly smaller, and hence insignificant, compared to its value in IL$^0$QG. The interpretation of $\psi_{\sigma^2}$ is that of a rescaled instantaneous buoyancy frequency deviation from the background frequency, $N$. The layer thickness deviation from the mean thickness, $h-H$, is proportional to the difference $\psi - R^2A(a_1,a_2)$. Finally, by the thermal-wind balance, the IL$^{(0,1)}$QG velocity includes, implicitly, vertical curvature, proportinal to $\psi_{\sigma^2}$.

The behavior of IL$^0$QG and IL$^{(0,1)}$QG in the problem of baroclinic instability was recently investigated in some depth in \cite{Beron-Olascoaga-25}, where additional details about the models can be found.

\subsection{Conservation laws}

As one might expect, the Hamiltonian \eqref{eq:H} serves as the generator of infinitesimal time shifts $t \mapsto t + \varepsilon$. This is evident because it implies $(q,a) \mapsto (q,a) + \varepsilon\partial_t(q,a)$, i.e., $\delta_{\mathcal H}(q,a) = \partial_t(q,a)$. Hence, specializing \eqref{eq:inf} to $f[q,a] = (q,a)$, we have
\begin{equation}
    \mathbf J\frac{\delta \mathcal H}{\delta (q,a)} = \partial_t(q,a),
\end{equation}
which are Hamilton's equations. According to Noether's theorem, $\mathcal H$ is a conserved quantity since
\begin{equation}
    \delta_{\mathcal H}\mathcal H = \varepsilon \int_\Omega\frac{\delta\mathcal H}{\delta q}\partial_t q + \frac{\delta\mathcal H}{\delta a}\partial_t a \,dxdy = \varepsilon\{\mathcal H,\mathcal H\} \equiv 0,
\end{equation}
by the skew-symmetry of the Poisson bracket. However, through direct evaluation, when the flow domain is bounded by a solid coast, one finds
\begin{equation}
    \dot{\mathcal{H}} = \left.\psi\right\vert_{\partial \Omega}\dot\gamma + \int_\Omega \psi J(\psi,q - 2 A)\,dxdy,
    \label{eq:dotH}
\end{equation}
where the fact that \eqref{eq:advQG-b} implies
\begin{equation}
    \partial_tA + J(\psi,A) = 0,
    \label{eq:dAdt}
\end{equation}
and consequently
\begin{equation}
    \partial_t(q-A) + J(q - 2A) = 0,
    \label{eq:qminusA}
\end{equation}
was used. The integral in \eqref{eq:dotH} vanishes due to the assumption of no-normal flow through $\partial \Omega$, as stated in Assumption \ref{ass:Omega}, along with $\nabla^\perp\cdot\nabla \equiv 0$. The conservation of $\mathcal{H}$ justifies Assumption \ref{ass:gamma}, regarding the preservation of circulation along $\partial \Omega$, on the left-hand side of \eqref{eq:gamma}. This is necessary for complete self-consistency when the deformation radius $R$ is finite. By contrast, when $R \uparrow \infty$, implying that the bottom of the active layer is horizontally rigid, one computes 
\begin{equation}
    \dot\gamma = \frac{d}{dt}\int_\Omega \partial_t(q-A-f)\,dxdy = - \int_\Omega J(\psi,q - 2A)\,dxdy  = 0.
\end{equation}
In other words, $\dot\gamma = 0$ is a consequence of advQG dynamics in that limiting case.  

Additional symmetry-related conservation laws are the momenta. Consider the case where the flow domain $\Omega$, lying on a $\beta$-plane, represents a zonal channel, which can be either infinite or periodic.  Then, the \emph{zonal momentum}, given by
\begin{equation}
    \mathcal M^x := -\int_\Omega yq\,dxdy,
\end{equation}
both generates infinitesimal zonal translations $x \mapsto x + \varepsilon$, since
\begin{equation}
    \mathbb J\frac{\delta\mathcal M^x}{\delta(q,a)} = 
    \begin{pmatrix}
        J(q,y) = \partial_xq\\
        J(a,y) = \partial_x a
    \end{pmatrix},
\end{equation}
and is preserved under dynamics, because
\begin{align}
    \delta_{\mathcal M^x}\mathcal H &= \varepsilon\int_\Omega\frac{\delta\mathcal H}{\delta q}\partial_xq + \frac{\delta\mathcal H}{\delta a}\partial_xa\,dxdy\nonumber\\ &= \varepsilon\int_\Omega \psi\partial_x(R^{-2}\psi -\nabla^2\psi - A - \beta y) + \psi\partial_aA\partial_xa\,dxdy\nonumber\\ &= \frac{\varepsilon}{2}\int_\Omega \partial_x(R^{-2}\psi^2 -\nabla^2\psi^2)\,dxdy
\end{align}
vanishes identically under a sufficient fall-off in $x$ or periodicity in that direction. Consider now the case in which $\Omega$ is axisymmetric and lies on the $f$-plane.  The conserved generator of infinitesimal rotations $\vartheta \mapsto \vartheta + \varepsilon$, is the \emph{angular momentum}, namely,
\begin{equation}
    \mathcal M^\vartheta := -\frac{1}{2}\int_\Omega (x^2+y^2)q\,dxdy = -\frac{1}{2}\int_\Omega r^3q\,drd\vartheta,
\end{equation}
where $r = \sqrt{x^2 + y^2}$ is the radial coordinate. In fact,
\begin{equation}
    \mathbb J\frac{\delta\mathcal M^\vartheta}{\delta(q,a)} = 
    \begin{pmatrix}
        J(q,r) = x\partial_yq - y\partial_xq = \partial_\vartheta q\\
        J(a,r) = x\partial_ya - y\partial_xa = \partial_\vartheta a
    \end{pmatrix}
\end{equation}
and further
\begin{align}
    \delta_{\mathcal M^\vartheta}\mathcal H &= \varepsilon\int_\Omega r\left(\frac{\delta\mathcal H}{\delta q}\partial_\vartheta q + \frac{\delta\mathcal H}{\delta a}\partial_\vartheta a\right)\,drd\vartheta\nonumber\\ &= \varepsilon\int_\Omega r\psi\partial_\vartheta(R^{-2}\psi -\nabla^2\psi - A) + r\psi\partial_aA\partial_\vartheta a\,drd\vartheta\nonumber\\ &= \frac{\varepsilon}{2}\int_\Omega \partial_\vartheta (R^{-2}r\psi^2 - r\nabla^2\psi^2)\,drd\vartheta
\end{align}
equals zero by periodicity in $\vartheta$, where $\nabla^2 = \partial_{rr} + r^{-1}\partial_r + r^{-2}\partial_{\vartheta\vartheta}$. 

The Casimirs of the Poisson bracket \eqref{eq:PB}, whose Poisson operator is given by \eqref{eq:J}, depend on $\alpha$ \cite{Beron-21-POFb}: if $\alpha = 0$
\begin{subequations}\label{eq:C}
\begin{equation}
    \mathcal C_{\Phi_1,\Phi_2}[q,a] = \int_\Omega q\Phi_1(a_1) + \Phi_2(a_1)\,dxdy\,\,\forall \Phi_1,\Phi_2,
    \label{eq:C-0}
\end{equation}
while if $\alpha > 0$
\begin{equation}
    \mathcal C_{c,\Phi_3}[q,a] = \int_\Omega cq + \Phi_3(a)\,dxdy\,\,\forall c = \mathrm{const}, \Phi_3.
    \label{eq:C-g0}
\end{equation}
\end{subequations}

In addition, the advQG dynamics equations \eqref{eq:advQG} preserve, when $\alpha > 0$,
\begin{equation}
    \mathcal W_{\Phi_4,\Phi_5}[q,a] := \int_\Omega q\Phi_4\big(A(a)\big) + \Phi_5\big(A(a)\big)\,dxdy\,\,\forall \Phi_4,\Phi_5.
    \label{eq:W}
\end{equation}
The proof of the conservation law $\dot{\mathcal W}_{\Phi_4,\Phi_5} = 0$ is a straightforward adaption of that given in \cite{Beron-Olascoaga-25} for the IL$^{(0,1)}$QG, which has into account \eqref{eq:dAdt}. The integral of motion \eqref{eq:W} is not a Casimir of the Poisson bracket \eqref{eq:PB}, but rather of this bracket with $a$ replaced by $A(a)$. Together with the Hamiltonian \eqref{eq:H}, it leads to a set of evolution equations comprising \eqref{eq:advQG}b and \eqref{eq:dAdt}. The resulting dynamics is such that the evolution of potential vorticity is not influenced by the details of the evolution of the individual advected quantities. Following \cite{Beron-Luesink-25}, we refer to \eqref{eq:W} as a weak Casimir. Together with the Casimirs \eqref{eq:C},  we can relate it to particle relabeling symmetry via Noether's theorem, as done in \cite{Beron-Luesink-25}.

\subsection{Production of Kelvin circulation}

Following \cite{Beron-Luesink-25}, the advQG dynamics equations \eqref{eq:advQG} follow as an Euler--Poincar\'e equation with advected quantities \cite{Holm-etal-02}, viz.,
\begin{equation}
    \partial_t\frac{\delta\mathcal L}{\delta\psi} + J\left(\psi,\frac{\delta\mathcal L}{\delta\psi}\right) = J\left(\frac{\delta\mathcal L}{\delta a},a\right)
    \label{eq:EP}
\end{equation}
with \emph{Lagrangian} defined by
\begin{equation}
    \mathcal L[\psi, a] := \frac{1}{2}\int_\Omega |\nabla\psi|^2 + R^{-2}\psi^2 - 2\psi\big(A(a) + \beta y\big)\,dxdy,
    \label{eq:L}
\end{equation}
and the advection equation for $a$, given by \eqref{eq:advQG-b}. We refer to \eqref{eq:EP} and \eqref{eq:advQG-b}, with Lagrangian \eqref{eq:L}, as the \emph{Euler--Poincar\'e equations for advQG dynamics}. Note that
\begin{equation}
    \frac{\delta\mathcal L}{\delta\psi} = - q,\quad
    \frac{\delta\mathcal L}{\delta a} = -\psi\partial_aA.
\end{equation}
Plugging these expressions into \eqref{eq:EP} leads to \eqref{eq:advQG-a}, as claimed.  Let $\Omega_t \subset \Omega$ be a fluid region at time $t$.  Integrating \eqref{eq:EP} over $\Omega_t$, one obtains, after changing variables,
\begin{equation}
    \frac{d}{dt}\int_{\Omega_t}\frac{\delta \mathcal L}{\delta\psi}\,dxdy = \int_{\Omega_t}J\left(\frac{\delta \mathcal L}{\delta a},a\right)\,dxdy
    \Longleftrightarrow
    \frac{d}{dt}\int_{\Omega_t}q\,dxdy = \int_{\Omega_t}J\big(\psi, A(a)\big)\,dxdy,
    \label{eq:K}
\end{equation}
where \eqref{eq:Ja} was used to obtain the right-hand side equality. Following \cite{Beron-Luesink-25}, let $\mathbf f(\mathbf x)$ be the Coriolis parameter's covector potential coefficient.  That is, $\nabla^\perp\cdot\mathbf f = f$. Next, define $\mathbf u(\mathbf x,t) := \nabla^\perp\psi(\mathbf x,t)$, i.e., the velocity vector coefficient, and consider $\mathbf q(\mathbf x,t)$, defined by
\begin{equation}
    \mathbf q := \mathbf u + \mathbf f - \nabla^\perp\nabla^{-2}(R^{-2}\psi - A), 
\end{equation}
where $\nabla^{-2}$ (sometimes also denoted $\Delta^{-1}$) must be understood in terms of the Green function associated with the elliptic problem \eqref{eq:q}. Note that $\nabla^\perp\cdot\mathbf q = q$, so $\mathbf q$ can be interpreted as the coefficient of the potential vorticity's covector potential. Finally, define 
\begin{equation}
    \mathcal K_{\Omega_t} := \int_{\Omega_t}q\,dxdy = \oint_{\partial \Omega_t} \mathbf q\cdot d\mathbf x,
\end{equation}
with the equality holding by Stokes theorem.  This generalizes the notion of \emph{Kelvin circulation} in 2D barotropic dynamics, that is, when $R^{-2} = 0$ and there are no advected quantities, implying $\mathbf q = \mathbf u + \mathbf f$, i.e., the coefficient of the absolute velocity covector, to the case of advQG dynamics. Then, applying Stokes theorem on the left-hand side of the equality on the right of \eqref{eq:Ja}, we have
\begin{equation}
    \dot{\mathcal K}_{\Omega_t} = \int_{\Omega_t}J\big(\psi,A\big)\,dxdy.
    \label{eq:dKdt}
\end{equation}
We refer to \eqref{eq:dKdt}, or more generally to \eqref{eq:K}, as the \emph{Kelvin circulation theorem for advQG dynamics}.  If advQG dynamics are specialized to IL$^{(0,\alpha)}$QG dynamics, the \emph{production of Kelvin circulation can be attributed to the misalignment between the gradients of dynamic topography and temperature}. 

\begin{remark}\label{rem:vol}
    If $\Omega_t$ is taken to be a fixed domain $\Omega$ bounded by a solid boundary, the right-hand side of \eqref{eq:dKdt} vanishes due to the requirement of zero normal to $\partial \Omega$, as imposed in Assumption \ref{ass:Omega}. Consequently, $\dot{\mathcal K}_{\Omega} = \dot\gamma - R^{-2}\frac{d}{dt}\int_\Omega \psi\,dxdy = 0$. Since, $\dot\gamma = 0$ by Assumption \ref{ass:gamma}, it follows that $\int_\Omega \psi\,dxdy = \mathrm{const}$, which admits the interpretation of volume preservation \textup{\cite{Beron-Luesink-25, Beron-Olascoaga-25}}.
\end{remark}

Two important observations must be considered regarding the ramifications of the absence of Kelvin circulation conservation by advQG dynamics. 
\begin{itemize}
    \item In recent years, the above result was used to offer an explanation within the realm of 2D QG modeling for remote sensing images of the ocean surface \cite{Holm-etal-21, Beron-21-POFa, Beron-21-POFb, Beron-24-POFa}. These images frequently reveal submesoscale roll-up vortices forming along thermal fronts. These roll-up vortices have been described as 2D ageostrophic thermal instabilities \cite{Gouzien-etal-17} or as fully 3D nonhydrostatic frontogenesis \cite{McWilliams-21}. This conveys added value to advQG modeling, extending beyond theoretical significance to practical applications.  Moreover, it potentially leads to more accurate motion types resulting from discovering a Lagrangian (or Hamiltonian) from data, opening an intriguing avenue for data-driven modeling \cite{Brunton-Kutz-24} tailored to preserve geometric structure.
    
    \item On the other hand, the generation of Kelvin circulation in advQG dynamics is linked to the lack of conservation of potential vorticity along fluid particle trajectories, which contrasts with the 3D QG theory for continuously stratified fluid \cite{Pedlosky-87}. In the 1990s, this aspect of advQG dynamics was not considered a beneficial feature but rather a deficiency. It was linked, in the IL$^0$QG model, to the existence of a ``forced compensating mode'' where layer thickness and buoyancy rearrange without modifying the energy, allowing for the possibility of unlimited growth of perturbations to a state of no motion \cite{Ripa-JGR-96, Ripa-JFM-95, Ripa-DAO-99}.  A similar mode with analogous consequences was identified in the IL$^{(0,1)}$QG model \cite{Beron-Olascoaga-25}.
\end{itemize}

\begin{remark}\label{rem:EP}
    This remark is an adaptation of \textup{\cite{Beron-Luesink-25}}. Geometrically, $\psi$ is viewed as an element of $\mathfrak{sdiff}(\Omega)$, and Euler--Poincar\'e system, formed by the Euler--Poincar\'e equation \eqref{eq:EP} with advected quantities \eqref{eq:advQG-b}, defined on $\mathfrak{sdiff}(\Omega) \times C^\infty(\Omega)^{\alpha+1}$. This follows from the constrained Hamilton principle
    \begin{equation}
        \delta\int_{t_0}^{t_1} \mathcal L[\psi,a]\,dt = 0 : \delta\psi = \partial_t\eta + J(\psi,\eta),\,\,\delta a = J(\eta,a),
    \end{equation}
    where $\eta(\mathbf x,t)$ is arbitrary. The solution leads to 
    \begin{equation}
        J^\eta := \left\langle \frac{\delta \mathcal L}{\delta \psi},\eta\right\rangle = \mathrm{const}
    \end{equation}
    along the motion produced by the Euler--Poincar\'e system \eqref{eq:EP} and \eqref{eq:advQG-b}. By direct computation,
    \begin{equation}
        \dot J^\eta = \int_\Omega \frac{\delta \mathcal L}{\delta q}{\delta q} +  \frac{\delta \mathcal L}{\delta a}\delta a\, dxdy \equiv \delta \mathcal L.
    \end{equation}
    Therefore, $J^\eta$ is conserved iff the Lagrangian $\mathcal L$ is symmetric.  This \textbf{Noether's theorem on the Lagrangian side}.  We can refer to $J^\eta$ as a \textbf{Noether quantity}. 
    
    Let $\boldsymbol\ell$ denote a 2-tuple array of fluid particle \textup{(}i.e., Lagrangian\textup{)} labels, which is viewed as the coefficient of a vector in the tangent bundle to the configuration space, $\Omega$, denoted $T\Omega$.  Let $\mathbf p$ be the canonical momentum conjugate to $\boldsymbol\ell$. Then $(\boldsymbol\ell,\mathbf p)$ belongs to the contangent bundle to $\Omega$, $T^*\Omega$. The duality of Euler--Poincar\'e formulation, \eqref{eq:EP} and \eqref{eq:advQG-b}, and the Hamiltonian formulation, with Lie--Poisson bracket \eqref{eq:LP} or, equivalently, \eqref{eq:PB}, is achieved via a partial Legendre transform $(\psi,a) \mapsto (q,a)$ constructed using the \textbf{momentum map} $m : T^*\Omega \to \mathfrak{sdiff}(\Omega)^*$, given by
    \begin{equation}
        m(\boldsymbol\ell,\mathbf p) \equiv J(\mathbf p,\boldsymbol\ell) = \frac{\delta \mathcal L}{\delta \psi} = -q \Longleftrightarrow J^\eta = \langle J(\mathbf p,\boldsymbol\ell),\eta\rangle,
        \label{eq:m}
    \end{equation}
    that is,
    \begin{equation}
        \mathcal H[q,a] = \langle J(\mathbf p,\boldsymbol\ell),\psi\rangle - \mathcal L[\psi,a],
    \end{equation}
    which is furthermore such that
    \begin{equation}
        \frac{d}{dt}\int_{\Omega_t} J(\mathbf p,\boldsymbol\ell)\,dxdy = \int_{\Omega_t} J\Big(\frac{\delta \mathcal L}{\delta a},a\Big)\,dxdy,
    \end{equation}
    which is the left-hand side equality in \eqref{eq:K}. 
    
    To understand \eqref{eq:m}, it is useful to invoke the Clebsch variational principle with two Lagrange multipliers, one enforcing advection of $a$ and the another one, $\mathbf p$, enforcing the action of $\mathrm {SDiff}(\Omega)$ on $T\Omega$, which results in advection of $\boldsymbol\ell$.  Extended details are given in \textup{\cite{Beron-Luesink-25}}, following \textup{\cite{Holm-24}}. The above exposition explicitly reveals that the \textup{(}minus\textup{)} Jacobian is the momentum map, a point that was not emphasized in \textup{\cite{Beron-Luesink-25}}.
    
    Finally, since $J^\eta$ represents a Noether quantity, it is appropriate to regard the Kelvin circulation theorem as the \textbf{Kelvin--Noether circulation theorem}.  In \textup{\cite{Beron-Luesink-25}} we make an explicit connection of $J^\eta = \mathrm{const}$ with Casimir conservation as a consequence the invariance of Eulerian variables under fluid particle relabeling.
\end{remark}

\section{Irreversible advQG dynamics}\label{sec:il-irreversible}

\subsection{The nonsymmetric case}

We begin constructing irreversible advQG dynamics, hereafter referred to as \emph{advQGi dynamics}, for the case without symmetry. Specifically, we assume that $\Omega$ spans the infinite $f$-plane or represents a possibly multiply connected domain thereof, enforcing Assumption \ref{ass:Omega} in every case. Let 
\begin{equation}
    J(\mathbf x,t)\in C^\infty(\Omega)^{n(\alpha)},\quad 
    n(\alpha) :=  
    \begin{cases}
        \frac{1}{2}\alpha + 1 & \text{if $\alpha$ : even,}\\
        \frac{1}{2}(\alpha + 1) & \text{if $\alpha$ : odd.}
    \end{cases}
\end{equation}
Then, define
\begin{equation}
    \hat{\mathbb J} := 
    \begin{bmatrix}
    0 & \operatorname{diag}J\\
    -\operatorname{diag}J & 0
    \end{bmatrix},
\end{equation}
and build with it the skew-adjoint operator depending on the parity of $\alpha$:
\begin{equation}
    \mathbb J := 
    \begin{cases}
        \hat{\mathbb J} & \text{if $\alpha$ is even,}\\
        \hat{\mathbb J} \oplus (0) & \text{if $\alpha$ is odd.}
    \end{cases}
    \label{eq:J2}
\end{equation}
Plugging \eqref{eq:J2} in \eqref{eq:bra2} produce, for any $f,g[u] \in C^\infty(\mathfrak F_\Omega^{\alpha+2})$ with $u = (q,a)$, the following skew-symmetric bracket depending on $\alpha$'s parity:
\begin{equation}
    \{\hspace{-.125cm}\{f,g\}\hspace{-.125cm}\} = \sum_{j=1}^{n(\alpha)}\int_\Omega J_j\left(\frac{\delta f}{\delta u_j}\frac{\delta g}{\delta u_{j + n(\alpha)}} - \frac{\delta f}{\delta u_{j + n(\alpha)}}\frac{\delta g}{\delta u_j}\right)\,dxdy.
    \label{eq:bra2-advQG}
\end{equation}
It is easy to see that the above bracket also satisfies the Jacobi identity. Define
\begin{equation}
    \lambda := \{\hspace{-.125cm}\{\mathcal C,\mathcal H\}\hspace{-.125cm}\}
\end{equation}
where $\mathcal H$, the Hamiltonian for reversible advQG dynamics, is given in \eqref{eq:H} and $\mathcal C$ is a Casimir of the corresponding Poisson bracket, as given by \eqref{eq:PB}. Depending on $\alpha$, this can be chosen from \eqref{eq:C}. Using the skew-symmetric bracket specified in \eqref{eq:bra2-advQG}, the metriplectic induced two-bracket, as defined in \eqref{eq:two}, evaluates to the following when considering $f = \mathcal{C}$, $g = u_i$, $i = 1,2,\dotsc,\alpha + 2$, and $h = \mathcal{H}$:
\begin{equation}
    (\mathcal C,u_i)_{\mathcal H} = 
    \begin{cases}
        \lambda J_i\frac{\delta\mathcal H}{\delta u_{i+n(\alpha)}} & \text{if } i = 1,2,\dotsc,n(\alpha),\\
        -\lambda J_{i-n(\alpha)}\frac{\delta\mathcal H}{\delta u_{i-n(\alpha)}} & \text{if } i = n(\alpha)+1,n(\alpha)+2,\dotsc,2n(\alpha),\\
        0 & \text{if } i = 2n(\alpha) + 1,
    \end{cases}
    \label{eq:two-advQG}
\end{equation}
where the third possibility is included when $\alpha$ is odd. Lastly, using Proposition \ref{pro:nonsym}, the following set of equations controlling advQGi dynamics with respect to internal energy $\mathcal U = \mathcal H$ and entropy $\mathcal S = \mathcal C$, result:
\begin{subequations}
\begin{align}
   \partial_t q &= \{q,\mathcal H\} + (\mathcal C,q)_{\mathcal H} = - J\big(\psi, q - A\big) + \lambda\psi J_1 \partial_{a_{n(\alpha)}}A,\label{eq:iadvQG-q}\\
   \partial_t a_i &= \{a_i,\mathcal H\} + (\mathcal C,a_i)_{\mathcal H} = - J(\psi,a_i) + (\mathcal C,a_i)_{\mathcal H},
   \label{eq:iadvQG-a}
\end{align} 
where
\begin{equation}
    (\mathcal C,a_i)_{\mathcal H} = 
    \begin{cases}
        \lambda \psi J_{i+1}\partial_{a_{i+n(\alpha)}}A & \text{if } i = 1,2,\dotsc,n(\alpha)-1,\\
        \lambda \psi J_{i-n(\alpha)+1} & \text{if } i = n(\alpha),\\
        -\lambda \psi J_{i-n(\alpha)+1}\partial_{a_{i-n(\alpha)}}A & \text{if } i = n(\alpha)+1,n(\alpha)+2,\dotsc,2n(\alpha)-1,\\
        0 & \text{if } i = 2n(\alpha),    
    \end{cases}
    \label{eq:advQGi-gen}
\end{equation}
\label{eq:advQGi}
\end{subequations}
where the forth possibility is included when $\alpha$ is odd. 

For future reference and to illustrate the evaluation of the general expression \eqref{eq:advQGi-gen}, we consider two relevant specific cases, $\alpha = 0$ and 1.  Consider first the case $\alpha = 0$, which corresponds to $n(\alpha) = \frac{1}{2}\alpha + 1 = 1$.  The phase space variables in this case are $(q,a_1)$. The equations for advQGi dynamics are:
\begin{subequations}\label{eq:advQGi-0}
\begin{align}
    \partial_t q + J\big(\psi, q - A(a_1)\big) &= \lambda\psi J_1 A'(a_1),\label{eq:advQGi-0-q}\\
    \partial_t a_1 + J(\psi,a_1) &= \lambda\psi J_1,
\end{align}
where
\begin{equation}
    \lambda =  \{\hspace{-.125cm}\{\mathcal C_{\Phi_1,\Phi_2},\mathcal H\}\hspace{-.125cm}\} = \int_\Omega \psi J_1\big(\Phi_1'(a_1)A'(a_1) + q\Phi_1'(a_1) + \Phi_2'(a_1)\big)\,dxdy,
    \label{eq:lambda-0}
\end{equation}
\end{subequations}
where $\mathcal C_{\Phi_1,\Phi_2}$ is the most general Casimir of the Poisson bracket for the corresponding reversible advQG dynamics, given in \eqref{eq:C-0}. Consider now the case the case $\alpha = 1$. This also has $n(\alpha) =  \frac{1}{2}(\alpha + 1) = 1$, but the phase space variables are $(q,a_1,a_2)$. The equations for advQGi dynamics are:
\begin{subequations}\label{eq:advQGi-1}
\begin{align}
    \partial_t q + J\big(\psi, q - A(a_1,a_2)\big) &= \lambda\psi J_1 \partial_{a_1}A,\label{eq:advQGi-1-q}\\
    \partial_t a_1 + J(\psi,a_1) &= \lambda\psi J_1,\\
    \partial_t a_2 + J(\psi,a_2) &= 0,
\end{align}
where
\begin{equation}
    \lambda =  \{\hspace{-.125cm}\{\mathcal C_{c,\Phi_3},\mathcal H\}\hspace{-.125cm}\} = \int_\Omega \psi J_1\left(c\partial_{a_1}A + \partial_{a_1}\Phi_3\right)\,dxdy,
\end{equation}
\end{subequations}
where $\mathcal C_{c,\Phi_3}$ is the most general Casimir of the Poisson bracket for the corresponding reversible advQG dynamics, given in \eqref{eq:C-g0}.

\subsubsection{Direct verification of irreversibility}

The metriplectic formulation developed in Section \ref{sec:formulation}, used to derive the equations of advQGi dynamics \eqref{eq:advQGi}, has been constructed to ensure irreversibility, namely, $\dot{\mathcal{U}} = \dot{\mathcal{H}} = 0$ and $\dot{\mathcal{S}} = \dot{\mathcal{C}} \ge 0$. It is reassuring that this can be directly verified, which we proceed to show. We begin with the first law of thermodynamics. Note that
\begin{align}
    \frac{DA}{Dt}  
    &= 
    \partial_aA\frac{Da}{Dt}  \nonumber\\
    &=
    \lambda\psi\left(\sum_{i=1}^{n(\alpha)-1}J_{i+1}\partial_{a_i}A\partial_{a_{i+n(\alpha)}}A + J_1\partial_{n(\alpha)}A - \sum_{i=n(\alpha)+1}^{2n(\alpha)-1}J_{i-n(\alpha)}\partial_{a_i}A\partial_{a_{i-n(\alpha)}}A \right)\nonumber\\
    &= 
    \lambda \psi J_1\partial_{a_{n(\alpha)}}A,
\end{align}
where
\begin{equation}
    \frac{D}{Dt} := \partial_t + J(\psi,\cdot\hspace{.05cm})
\end{equation}
denotes the total derivative.  The above readily implies \eqref{eq:qminusA}, which simplifies the evaluation of $\dot{\mathcal{H}}$ to that in the reversible advQG dynamics case. This results in $\dot{\mathcal{H}} = 0$, implying the preservation of $\mathcal{U} = \mathcal{H}$. $\square$

To verify the second law of thermodynamics, the case $\alpha = 0$ must be analyzed separately from the cases where $\alpha > 0$ because the reversible advQG dynamics have different Casimir structures. We begin with $\alpha = 0$ and time differentiate \eqref{eq:C-0}:
\begin{align}
    \dot{\mathcal C}_{\Phi_1,\Phi_2} 
    &= 
    \int_\Omega \partial_tq \Phi_1(a_1) + q\partial_t\Phi_1(a_1) +\partial_t\Phi_2(a_1)\,dxdy\nonumber\\
    &=
    \int_\Omega (\mathcal C_{\Phi_1,\Phi_2},q)_{\mathcal H}\Phi_1(a_1) + \big(q\Phi_1'(a_1) + \Phi_2'(a_1)\big) (\mathcal C_{\Phi_1,\Phi_2},a_1)_{\mathcal H}\,dxdy\nonumber\\
    &=
    \lambda\int_\Omega \{\hspace{-.125cm}\{q,\mathcal H\}\hspace{-.125cm}\}\Phi_1(a_1) + \big(q\Phi_1'(a_1) + \Phi_2'(a_1)\big) \{\hspace{-.125cm}\{a_1,\mathcal H\}\hspace{-.125cm}\}\,dxdy\nonumber\\
    &=
    \lambda \int_\Omega\{\hspace{-.125cm}\{q\Phi_1(a_1) +\Phi_2(a_1),\mathcal H\}\hspace{-.125cm}\}\,dxdy
    = 
    \lambda\{\hspace{-.125cm}\{\mathcal C_{\Phi_1,\Phi_2},\mathcal H\}\hspace{-.125cm}\}
    =
    \lambda^2,
\end{align}
where the derivation property of $\{\hspace{-.125cm}\{\,,\hspace{.05cm}\}\hspace{-.125cm}\}$ was used. This implies that $\mathcal S = \mathcal C_{\Phi_1,\Phi_2}$ is either produced or preserved, but not consumed over time. $\square$

Now, we turn to $\alpha > 0$ and take the time derivative of \eqref{eq:C-g0}:
\begin{align}
    \dot{\mathcal C}_{c,\Phi_3} 
    &= 
    \int_\Omega c\partial_tq + \partial_t\Phi_3(a)\,dxdy\nonumber\\
    &=
    \int_\Omega (\mathcal C_{c,\Phi_3},cq)_{\mathcal H} + (\mathcal C_{c,\Phi_3},\Phi_3(a))_{\mathcal H}\,dxdy\nonumber\\
    &=
    \lambda \int_\Omega\{\hspace{-.125cm}\{cq +\Phi_3(a),\mathcal H\}\hspace{-.125cm}\}\,dxdy
    = 
    \lambda\{\hspace{-.125cm}\{\mathcal C_{c,\Phi_3},\mathcal H\}\hspace{-.125cm}\}
    =
    \lambda^2,
\end{align}
where, as for $\alpha = 0$, the derivation property of $\{\hspace{-.125cm}\{\,,\hspace{.05cm}\}\hspace{-.125cm}\}$ was conveniently utilized.  Therefore, over time, $(\mathcal{S} = \mathcal{C}_{c,\Phi_3}$ is generated or conserved, but it cannot decrease. $\square$ 

\subsubsection{Kelvin circulation theorem for advQGi dynamics}

Note that \eqref{eq:advQGi} are obtained as
\begin{subequations}\label{eq:iEP}
\begin{align}
    \partial_t\frac{\delta \mathcal L}{\delta \psi} + J\left(\psi,\frac{\delta \mathcal L}{\delta \psi}\right) &= J\left(\frac{\delta \mathcal L}{\delta a},a\right) + \left(\frac{\delta \mathcal L}{\delta \psi},\mathcal C\right)_{\big\langle\frac{\delta \mathcal L}{\delta \psi}, \psi\big\rangle - \mathcal L}\label{eq:iEP-q}\\
    \partial_ta + J(\psi,a) &= (\mathcal C,a)_{\big\langle\frac{\delta \mathcal L}{\delta \psi}, \psi\big\rangle - \mathcal L}.\label{eq:iEP-a}
\end{align}
\end{subequations}
Here, $\mathcal L$ is the Lagrangian for reversible advQG dynamics, given by \eqref{eq:L}, and
\begin{equation}
    \left\langle\frac{\delta \mathcal L}{\delta \psi}, \psi\right\rangle - \mathcal L = \mathcal H
    \label{eq:LT}
\end{equation} 
is the corresponding Hamiltonian, given in \eqref{eq:H}, as it follows by integrating by parts. It is important to realize that, unlike the Euler--Poincar\'e equations for reversible advQG dynamics, given by \eqref{eq:EP} and \eqref{eq:advQG-b}, the set of equations \eqref{eq:iEP} do not follow from Hamilton's principle, as discussed in Remark \ref{rem:EP}, where \eqref{eq:LT} is interpreted as partial Legendre transform $(\psi,a) \mapsto (q,a)$.

However, a Kelvin circulation theorem, analogous to \eqref{eq:dKdt} for reversible advQG dynamics, can be derived by integrating \eqref{eq:iEP-q} over $\Omega_t$.  The result is
\begin{equation}
    \dot{\mathcal K}_{\Omega_t} = \int_{\Omega_t} J\big(\psi,A(a)\big) + \lambda\psi J_1 \partial_{a_{n(\alpha)}}A\, dxdy
    \label{eq:dKdt-iadvQG}
\end{equation}
where $\mathcal K_{\Omega_t}$ is defines as in \eqref{eq:K}, which follows after changing variables and using the  metriplectic induced two-bracket \eqref{eq:two-advQG} evaluated for $u_1 = q$. This represents the \emph{Kelvin circulation circulation for advQGi dynamics}.  Irreversibility can either increase or decrease the production of circulation along material loops.

\begin{remark}
   It is important to note that the concept of irreversibility, as defined in Definition \ref{def:reversible}, which pertains to the production of entropy while conserving internal energy, does not rule out the existence of a smoothly invertible map, i.e., a diffeomorphism, that connects initial fluid particle positions with their positions at a later time. This allows for the application of a change of variables to obtain the result \eqref{eq:dKdt-iadvQG}.
\end{remark}

\subsubsection{Specific forms of advQGi dynamics}

While irreversibility is ensured regardless of the specific type of advQGi dynamics, there is flexibility in the choice. This freedom arises from the ability to select the functions \eqref{eq:J}. Of all such functions, the choice of $J_1$ will have the most evident impact. 

\paragraph{Volume preserving advQGi dynamics.}

Consider the Kelvin circulation theorem for advQGi \eqref{eq:dKdt-iadvQG} with $\Omega_t$ taken to be a fixed domain $\Omega$ bounded by a solid boundary. Assuming simple connectivity for simplicity, we compute
\begin{equation}
    - R^{-2}\frac{d}{dt}\int_\Omega \psi\,dxdy = \int_\Omega \lambda\psi J_1 \partial_{a_{n(\alpha)}}A\,dxdy
\end{equation}
by Assumptions \ref{ass:Omega} and \ref{ass:gamma}.  As noted in Remark \ref{rem:vol}, the vanishing time derivative on the left-hand side of the above relationship indicates volume preservation. For advQGi, this can be enforced by setting $J_1 = 0$. Clearly, irreversibility is achieved under these conditions when $\alpha > 1$.

\paragraph{Quasi-conserving advQGi dynamics.}

Consider advQGi dynamics with $\alpha = 0$, with motion equations given by \eqref{eq:advQGi-0}. Assume irreversibility with respect to the following specific choice of entropy:
\begin{equation}
    \mathcal S = \mathcal C_{0,A} = \int_\Omega A(a_1)\,dxdy.
    \label{eq:S0}
\end{equation}
This implies
\begin{equation}
    \lambda = \{\hspace{-.125cm}\{\mathcal C_{0,A},\mathcal H\}\hspace{-.125cm}\} = \int_\Omega J_1 \psi A'(a_1)\,dxdy.
\end{equation}
Let $\varphi(\mathbf x,t) \in C^\infty(\Omega)$ such that
\begin{equation}
    \int_\Omega \varphi\,dxdy = 1.
\end{equation}
Then, consider the choice
\begin{equation}
    J_1 = \frac{J\big(A(a_1),\psi\big) + \varphi}{\psi A'(a_1)},
    \label{eq:J1}
\end{equation}
which implies
\begin{equation}
    \lambda = \int_\Omega J(A(a_1),\psi) + \varphi\,dxdy = \int_\Omega \varphi\,dxdy = 1.
    \label{eq:lambda}
\end{equation}
Plugging \eqref{eq:lambda} and \eqref{eq:J1} into \eqref{eq:advQGi-0}, it follows that 
\begin{subequations}
    \begin{align}
         \partial_tq + J(\psi,q) &= \varphi,\\
         \partial_tAa_1 +  J\big(\psi,a_1\big) &= \varphi/A'(a_1).
    \end{align} 
\end{subequations}
It is easy to verify irreversibility directly: $\dot{\mathcal H} = 0$ and $\dot{\mathcal C}_{0,A} = \int_\Omega \varphi\, dxdy = 1 > 0$.

Consider now advQGi dynamics with $\alpha = 1$, with equations of motion given by \eqref{eq:advQGi-1}. Assume that the system is irreversibility with respect to the specified choice of entropy as given below:
\begin{equation}
    \mathcal S = \mathcal C_{0,A} = \int_\Omega A(a_1,a_2)\,dxdy,
    \label{eq:S1}
\end{equation}
which conduces to
\begin{equation}
    \lambda = \{\hspace{-.125cm}\{\mathcal C_{0,A},\mathcal H\}\hspace{-.125cm}\} = \int_\Omega J_1 \psi\partial_{a_1}A\,dxdy.
\end{equation}
As above, it is evident that the choice
\begin{equation}
    J_1 = \frac{J\big(A(a_1,a_2),\psi\big) + \varphi}{\psi \partial_{a_1}A}
    \label{eq:J1-1}
\end{equation}
reduces \eqref{eq:advQGi-1} to
\begin{subequations}
    \begin{align}
         \partial_tq + J(\psi,q) &= \varphi,\\
         \partial_tAa_1 +  J\big(\psi,a_1\big) &= \varphi/\partial_{a_1}A,\\
         \partial_tAa_2 +  J\big(\psi,a_2\big) &= 0,
    \end{align} 
\end{subequations}
which imply $\dot{\mathcal H} = 0$ and $\dot{\mathcal C}_{0,A} = \int_\Omega \varphi\, dxdy = 1 > 0$, as expected.

Let us set $\varphi$ to a constant given by
\begin{equation}
    \varphi = \frac{1}{\int_\Omega dxdy},
\end{equation}
which can be very small if the flow domain is sufficiently large. When specialized to IL$^0$QG and IL$^{(0,1)}$QG dynamics, the results above imply irreversible thermal QG dynamics, both standard and stratified, with approximately conserved potential vorticity and buoyancy (temperature). This can address issues identified with thermal QG dynamics, which are associated with the lack of conservation of potential vorticity, by reinterpreting the dynamics as irreversible. Interestingly, in this specialization, the entropies given by \eqref{eq:S0} and \eqref{eq:S1} can be interpreted as \emph{actual} entropies, being integrals of buoyancy (temperature).

A final comment pertains to the effect of the specific type of irreversibility induced by the choices \eqref{eq:J1} for $\alpha=0$ and \eqref{eq:J1-1} for $\alpha = 1$. In either case, generation of Kelvin circulation along transported fluid loops is produced at a rate given by
\begin{align}
    \dot{\mathcal K}_{\Omega_t} = \int_{\Omega_t}\varphi\,dxdy,
\end{align}
which is independent of misalignments between the gradients of dynamic topography and buoyancy (temperature), provided that advQG dynamics is specialized as IL$^{(0,\alpha)}$QG dynamics. 

\paragraph{Forced--dissipative advQGi dynamics.}

Let $F(\mathbf x,t) \in C^\infty(\Omega)$ be a forcing term and $D(\mathbf x,t) \in C^\infty(\Omega)$ represent dissipation. For instance, these can take the specific forms:
\begin{equation}
    F = \nabla^\perp\cdot\boldsymbol\tau,\quad 
    D = - r\nabla^2\psi - \nu(p)\nabla^{2p}q.
    \label{eq:FplusD}
\end{equation}
Here, $\boldsymbol \tau(\mathbf x,t)$ represents the wind stress applied at the top of the active fluid layer. The coefficient $r > 0$ can be interpreted as the drag at the interface with the inactive fluid layer, and $\nu(p) > 0$ is a hyperviscosity for $p \ge 1$. The drag term is most meaningful when the fluid lies atop a rigid bottom, possibly with a free surface \cite{Smith-Vallis-02}.  This configuration can be represented by advQG dynamics by appropriately choosing the parameter $R$ in \eqref{eq:q} (or more specifically setting $g'$ to $g$, gravity, in its specialization to IL$^{(0,\alpha)}$QG dynamics).

Setting
\begin{equation}
    J_1 = \frac{F+D}{\psi\partial_{a_{n(\alpha)}}A},
\end{equation}
the evolution of potential vorticity is subjected to forcing and dissipation, in a thermodynamically consistent manner, according to
\begin{equation}
    \partial_t q +J\big(\psi, q - A(a)\big) = \lambda (F + D),
\end{equation}
which does not depend on the particular choice \eqref{eq:FplusD}.

A final remark is that $\lambda$ can be controlled in the particular, but relevant, advQGi dynamics with $\alpha = 0$ and $1$. With respect to the specific choices of entropy as given by \eqref{eq:S0} and \eqref{eq:S1}, when $\Omega$ is bounded, allowing for the possibility of multiple connections, the choices
\begin{equation}
    J_1 = \frac{1}{\psi A'(a)}\frac{F+D}{\sqrt{\big|\int_\Omega F+D\,dxdy\big|}},\quad J_1 = \frac{1}{\psi \partial_{a_1}A}\frac{F+D}{\sqrt{\big|\int_\Omega F+D\,dxdy\big|}},
\end{equation}
when $\alpha = 0$ and $1$, respectively, both result in $\lambda = \operatorname{sign}\int_\Omega F+D\,dxdy$. This adds further realism to the forcing and dissipation term $F+D$, particularly when $F$ and $D$ are taken to be of the form indicated in \eqref{eq:FplusD}.

\subsection{The symmetric case}

We finalize with the construction of advQGi dynamics by considering that the flow domain $\Omega$ either represents a zonal channel, periodic or infinitely long, on the $\beta$-plane, or has axial symmetry, lying on the $f$-plane, enforcing Assumption \ref{ass:Omega} in each case.

Using Proposition \ref{pro:sym}, advQGi dynamics with respect to $\mathcal U = \mathcal H_{\dot\chi}$ and $\mathcal S = \mathcal C$ are produced by \eqref{eq:advQGi} with the second possibility in \eqref{eq:advQGi-gen}, valid for $i=n(\alpha)$, replaced
\begin{equation}
    (\mathcal C,a_{n(\alpha)})_{\mathcal H_{\dot\chi}} = 
    \begin{cases}
        \lambda_{\dot\chi} J_1 (\psi - \dot\chi y) & \text{if $\Omega$ is zonally symmetric,}\\
        \lambda_{\dot\chi} J_1 (\psi - \tfrac{1}{2}\dot\chi r^2) & \text{if $\Omega$ is axisymmetric,}
    \end{cases}
\end{equation}
where
\begin{equation}
    \lambda_{\dot\chi} := \{\hspace{-.125cm}\{\mathcal C,\mathcal H_{\dot\chi}\}\hspace{-.125cm}\}.
\end{equation}
The interpretation of $\dot\chi$ changes from a constant velocity in the zonally symmetric case to a constant angular velocity in the axisymmetric cases. The choice of $\mathcal C$ depends on the value of $\alpha$.

The specializations discussed for the non-symmetric flow domain case apply to the symmetric case, and we will not repeat them here.

\section{A numerical test of metriplecticity}\label{sec:test}

We illustrate the metriplectic construction with advQG in the presence of prescribed topography. The topographic term is useful numerically because it promotes the roll-up of the buoyancy front, but it does not alter the Hamiltonian character of the model \cite{Holm-etal-21, Beron-21-POFa}. This enters the Hamiltonian as a prescribed, time-independent field, \(\psi_0(\mathbf x) \in C^\infty(\Omega)\), namely,
\begin{equation}
    \mathcal H[q,a]
    =
    \frac{1}{2}
    \int_\Omega
    |\nabla\psi|^2
    +
    R^{-2}\psi^2
    -
    2\psi_0(\mathbf x)A(a)
    \,dxdy ,
\end{equation}
with variational derivatives
\begin{equation}
    \frac{\delta\mathcal H}{\delta q}
    =
    -\psi,
    \qquad
    \frac{\delta\mathcal H}{\delta a}
    =
    (\psi-\psi_0)\partial_a A .
\end{equation}
Together with the Lie--Poisson bracket \eqref{eq:bra2-advQG}, topographically forced advQG dynamics are controlled by
\begin{equation}
    \partial_t q
    +
    J(\psi,q-A(a))
    =
    J(A(a),\psi_0),
\end{equation}
with \(a\) advected by the flow.

For the numerical simulations we consider the IL$^0$QG member of the family, which corresponds to \(\alpha=0\), \(a_1=\psi_\sigma\), and \(A(a_1)=R^{-2}\psi_\sigma\); cf.~Section~\ref{sec:ILQG}. Irreversibility is introduced with respect to the entropy
\begin{equation}
    \mathcal S
    =
    \mathcal C_{1,0}
    =
    \int_\Omega q\,dxdy ,
\end{equation}
which gives
\begin{equation}
    \lambda
    =
    \{\hspace{-.125cm}\{\mathcal C_{1,0},\mathcal H\}\hspace{-.125cm}\}
    =
    \int_\Omega
    J_1 R^{-2}(\psi-\psi_0)
    \,dxdy .
\end{equation}
Choosing
\begin{equation}
    J_1
    =
    \frac{R^2}{\psi-\psi_0}
    \frac{F}{\sqrt{|\int_\Omega F\,dxdy|}},
\end{equation}
with \(F(\mathbf x,t)\in C^\infty(\Omega)\), the metriplectic IL$^0$QG equations become
\begin{align}
    \partial_t q
    +
    J(\psi,q-R^{-2}\psi_\sigma)
    &=
    R^{-2}J(\psi_\sigma,\psi_0)
    +
    \operatorname{sign}
    \left(
    \int_\Omega F\,dxdy
    \right)F,
    \\
    \partial_t\psi_\sigma
    +
    J(\psi,\psi_\sigma)
    &=
    \operatorname{sign}
    \left(
    \int_\Omega F\,dxdy
    \right)
    \frac{R^2F\psi}{\psi-\psi_0}.
\end{align}

We solve these equations on the \(f\)-plane periodic domain \(\Omega=\mathbb R^2/(2\pi R\mathbb Z)^2\) using Zeitlin's \cite{Zeitlin-91} finite-dimensional approximation. A simple choice of \(F\) that serves our purposes is to take it constant; in particular, we set \(F/f_0R^2=0.1\), so that the metriplectic term provides a weak irreversible perturbation of the Hamiltonian evolution. The prescribed topographic streamfunction is \(\psi_0(x,y)/f_0R^2=5\sin(2x/R)\cos(2y/R)\), and the initial thermal streamfunction is chosen as the perturbed front \(\psi_\sigma(x,y,0)/f_0R^2=2\tanh(10y/R-5-2.5\cos(2x/R))\). The initial potential vorticity corresponds to initialization from a state of rest, namely, \(q(x,y,0)=R^{-2}\psi_\sigma(x,y,0)\).

Zeitlin's approximation replaces the Fourier representation of the fields by an \(N\times N\) matrix representation whose commutator bracket preserves the Lie--Poisson structure. The infinite family of Casimirs depending on arbitrary functions of the advected field \(\psi_\sigma\) is replaced by the corresponding finite set generated by traces of powers of the matrix representation of \(\psi_\sigma\). Thus the unperturbed calculation is an exactly Hamiltonian finite-dimensional system with discrete invariants. The metriplectic simulations reported below use the same structure-preserving discretization for the Hamiltonian dynamics and add the weak irreversible contribution described above. This provides a controlled numerical test of the effect of the metriplectic perturbation, while extending structure-preserving discretizations to the full metriplectic setting remains an open research direction.

Figure~\ref{fig:fields} compares the Hamiltonian and metriplectic evolutions. In both cases the nonlinear dynamics rolls up the initial buoyancy front into coherent vortical structures. Although introduced as a weak perturbation, the metriplectic contribution enhances the deformation and mixing of the front, producing a stronger rearrangement of the buoyancy field while retaining the large-scale coherent motion.

\begin{figure}[t!]
    \centering
    \includegraphics[width=\linewidth]{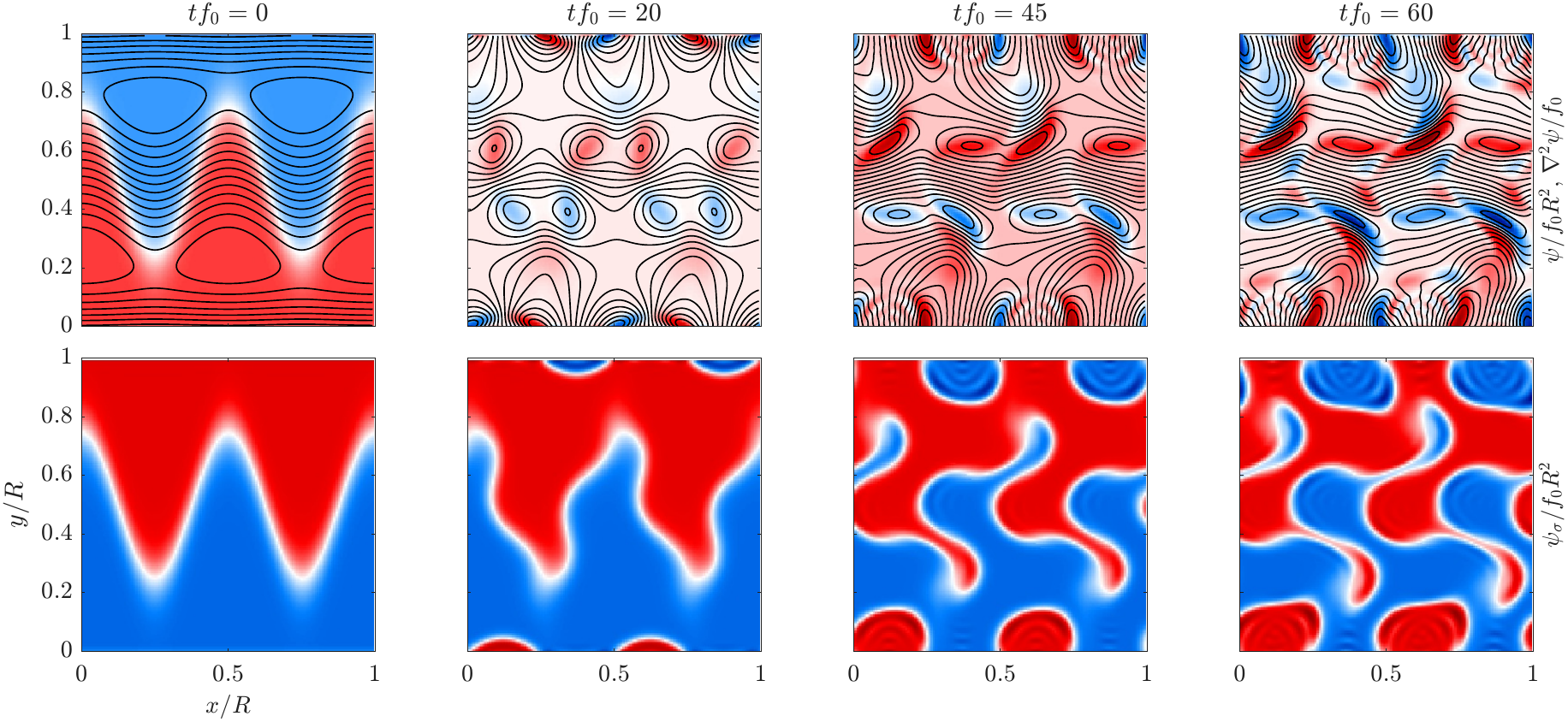}\\
    \includegraphics[width=\linewidth]{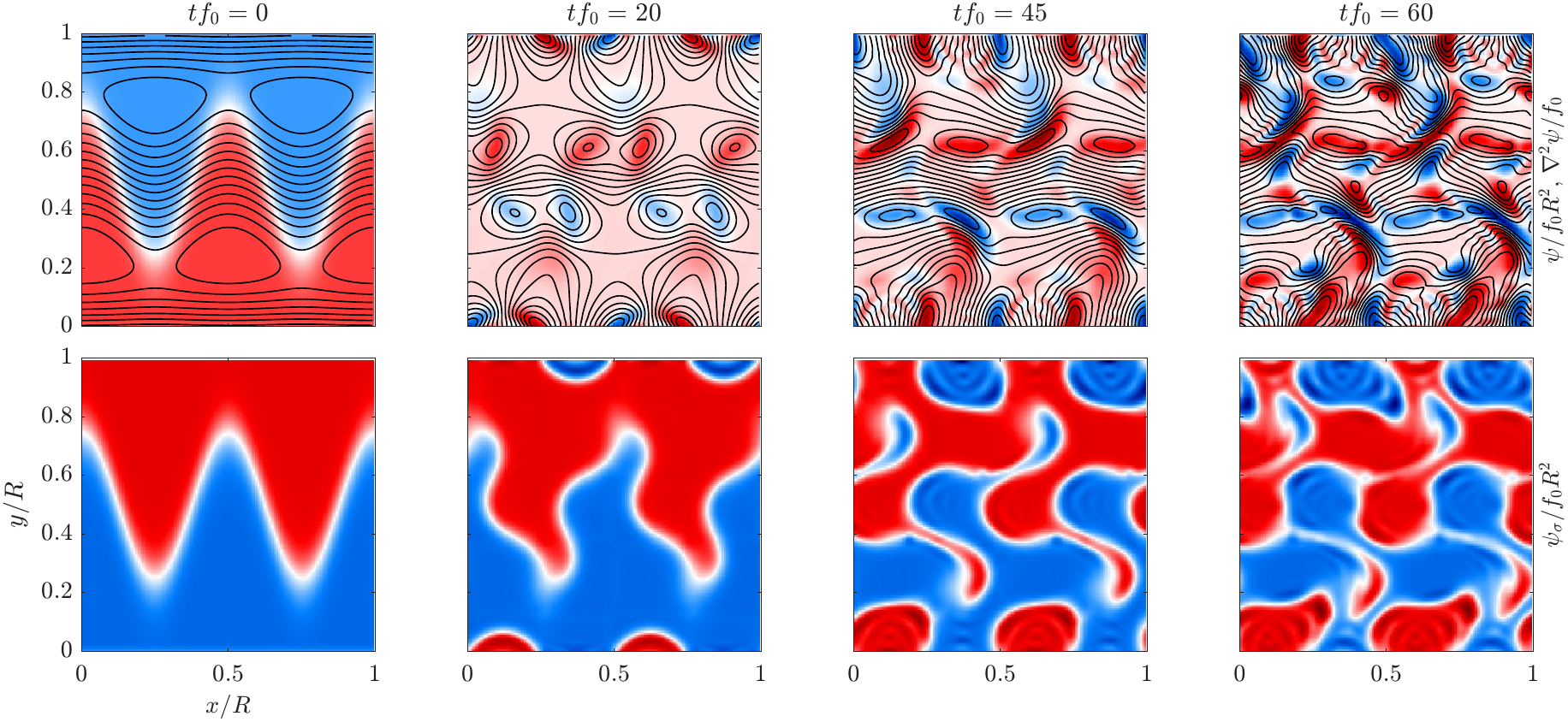}
    \caption{Evolution of the IL$^0$QG fields. Top block: Hamiltonian evolution computed with the structure-preserving Zeitlin discretization. Bottom block: evolution with the additional weak metriplectic contribution. In each block, the upper panels show relative vorticity \(\nabla^2\psi\) with streamlines \(\psi=\mathrm{const.}\) overlaid, and the lower panels show the buoyancy streamfunction \(\psi_\sigma\).}
    \label{fig:fields}
\end{figure}

Figure~\ref{fig:diagnostics} shows the corresponding diagnostics. The Hamiltonian Zeitlin system conserves the discrete invariants to numerical precision. In the metriplectic case the energy remains controlled while the Casimir invariants evolve, reflecting the selective irreversible rearrangement introduced by the metric bracket.

\begin{figure}[t!]
    \centering
    \includegraphics[width=.45\linewidth]{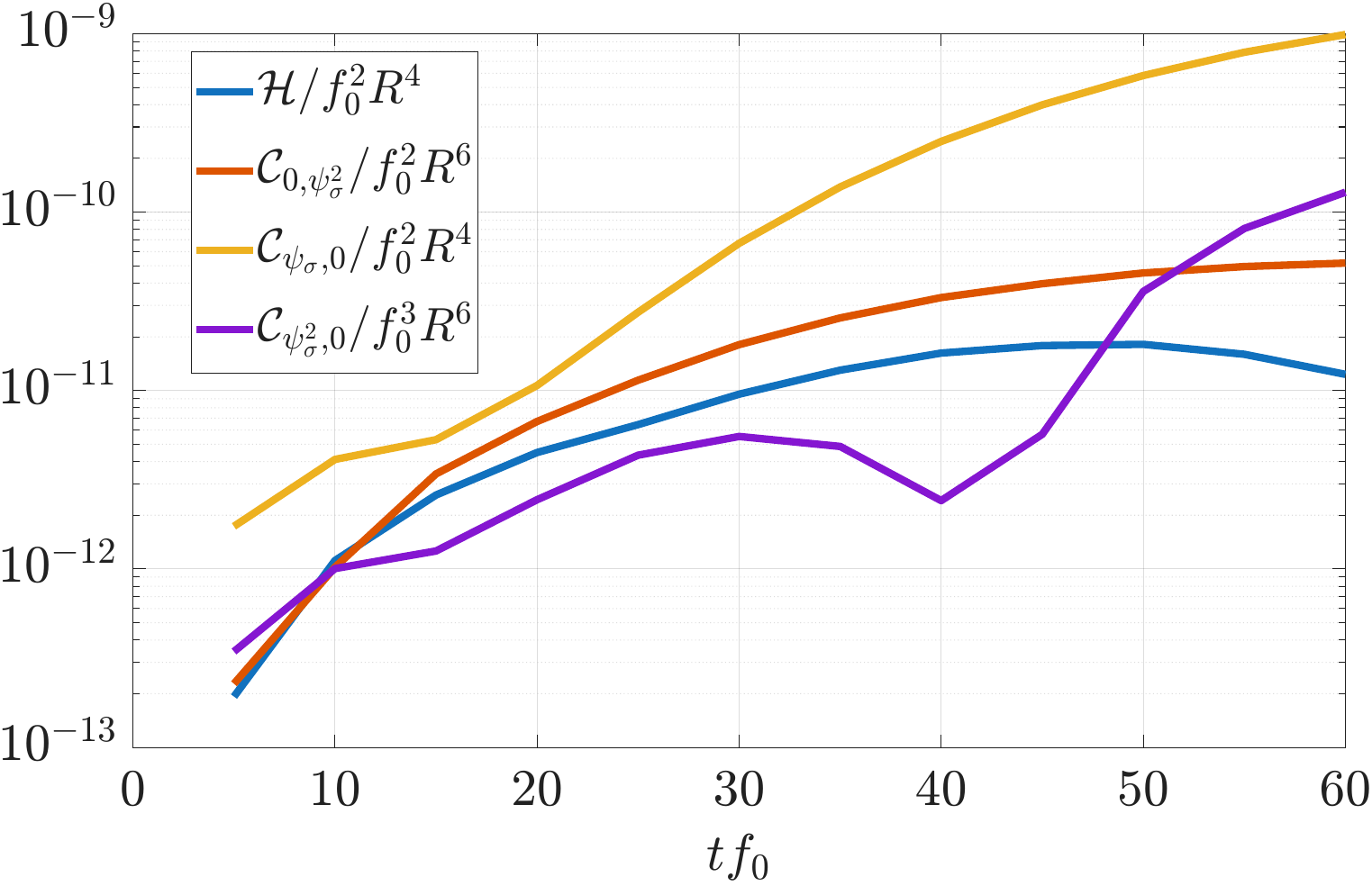}\,
    \includegraphics[width=.45\linewidth]{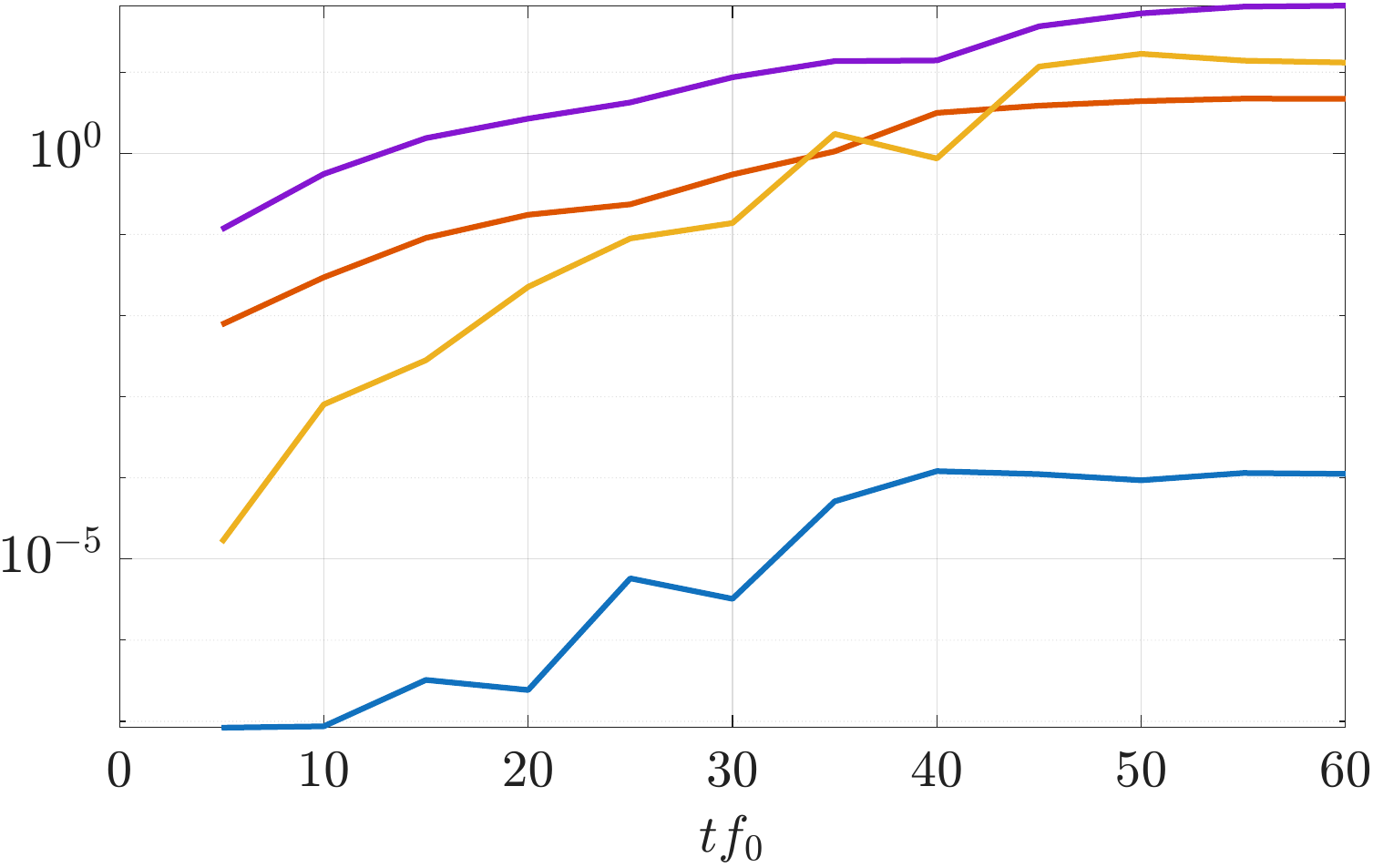}
    \caption{Diagnostic evolution for the Hamiltonian (left) and metriplectic (right) simulations, showing the relative changes of the Hamiltonian and representative discrete Casimir invariants.}
    \label{fig:diagnostics}
\end{figure}

\section{Conclusions}\label{sec:con}

In this paper, a general formulation of metriplectic dynamics is presented. The formulation is straightforward, involving the construction of the metriplectic four-bracket by multiplying two skew-symmetric brackets. This new formulation is then utilized to introduce thermodynamically consistent irreversibility in a generalized two-dimensional (2D) quasigeostrophic (QG) upper-ocean model with advected quantities. By design, the resulting dynamics ensure the conservation of internal energy and the generation of entropy, in accordance with the first and second laws of thermodynamics. The generalized 2D QG model with advected quantities includes the thermal QG model developed in the 1990s and the recently derived stratified thermal QG model. These models are characterized by producing Kelvin circulation along material loops, which may explain the tendency of thermal fronts to develop small-scale Kelvin–Helmholtz-like roll-up vortices. The production of Kelvin circulation in such models is associated with the lack of potential vorticity conservation, unlike the 3D QG model. The metriplectic dynamics formulation developed here enables the construction of specific types of irreversibility that can result, if desired, in dynamics that nearly preserve potential vorticity. Moreover, it allows for the incorporation of quite realistic forcing and dissipation.  The metriplectic construction was illustrated numerically for IL$^0$QG using a structure-preserving discretization of the Hamiltonian dynamics, showing the expected conservation of energy and irreversible evolution of the Casimirs. Further development of fully structure-preserving finite-dimensional metriplectic discretizations and applications of the proposed formulation beyond QG modeling provide natural directions for future investigation.

\section*{Acknowledgments}

F.J.B.V.\ thanks the Intercultural Outreach Initiative (IOI) in Isabela, Gal\'apagos, for their hospitality during the writing of a portion of this article. E.L.\ acknowledges financial support from NWO Grant No.\ VI.VIDI.213.070.

\section*{Author Declarations}

\subsection*{Conflict of Interest}

The authors have no conflicts of interest to disclose.

\subsection*{Author Contributions}

F.J.B.V.\ and E.L.\ contributed equally to the conceptual development and formulation of the theory. F.J.B.V.\ performed the numerical simulations with assistance from ChatGPT for code development and debugging. Both authors contributed to the analysis of the results and the writing of the manuscript.

\section*{Data and Software Availability}

The Julia and MATLAB codes used to produce the numerical results reported in this article are available from the authors upon reasonable request.

\appendix
\numberwithin{equation}{section}

\section{Irreversible HLQG dynamics}\label{app:qg}

The dynamics produced by the HLQG model develop on the \emph{invariant} subspace of advQG dynamics defined by $\{a = 0\}$. They are Hamiltonian \cite{Marsden-Weinstein-82}, with Hamiltonian given by \eqref{eq:H} where $q$ is given as in \eqref{eq:q}, except that $A(a) = 0$, and Poisson bracket given by \eqref{eq:PB} with $a = 0$.  The corresponding Casimir is 
\begin{equation}
    \mathcal C_\Phi = \int_\Omega\Phi(q)\,dxdy\quad\forall\Phi.
\end{equation}

Let $a(\mathbf x,t)$ be a \emph{passively} advected quantity.  The dynamics of $ u = (q,a)$ are Hamiltonian with the above Hamiltonian and (direct product, Lie--)Poisson bracket
\begin{equation}
    \{f,g\} := \int_\Omega qJ\left(\frac{\delta f}{\delta q},\frac{\delta g}{\delta q}\right) + aJ\left(\frac{\delta f}{\delta a},\frac{\delta g}{\delta a}\right)\,dxdy, 
\end{equation}
Its Casimir is given by
\begin{equation}
    \mathcal C_{\Phi_1,\Phi_2} = \int_\Omega \Phi_1(q) + \Phi_2(a)\,dxdy\quad\forall\Phi_1,\Phi_2.
\end{equation}

The evolution of $q$ is \emph{not} affected by $a$, even if the Hamiltonian
\begin{equation}
    \mathcal H[q,a] = \frac{1}{2}\int_\Omega |\nabla\psi|^2 + R^{-2}\psi^2 + 2a\,dxdy
\end{equation}
is considered. Assume that the flow domain does not have any symmetry, irreversible HLQG dynamics with respect to $\mathcal U = \mathcal H$ and $\mathcal S = C_{\Phi_1,\Phi_2}$ are constructed by choosing
\begin{equation}
    \{\hspace{-.125cm}\{f,g\}\hspace{-.125cm}\} = \int_\Omega J_1\left(\frac{\delta f}{\delta q}\frac{\delta g}{\delta a} - \frac{\delta f}{\delta a}\frac{\delta g}{\delta q}\right)\,dxdy
\end{equation}
for any $J(\mathbf x,t) \in C^\infty(\Omega)$. This leads to the following set motion equations
\begin{subequations}
    \begin{align}
        \partial_t q + J(\psi,q) &= \lambda J_1,\\
        \partial_t a + J(\psi,a) &= \lambda\psi J_1,
    \end{align}
\end{subequations}
where $\lambda := \{\hspace{-.125cm}\{\mathcal C_{\Phi_1,\Phi_2},\mathcal H\}\hspace{-.125cm}\}$. By direct computation, $\dot{\mathcal H} = 0$ and $\dot{\mathcal S} = \lambda^2 \ge 0.$ It is important to realize that the resulting equations are coupled through $\lambda$, which is a functional of $(q,a)$. In the symmetric flow domain case, which is irreversible with respect to $\mathcal{U} = \mathcal{H}_\chi$, the symmetry-induced Hamiltonian, the same equations produce the results except that on the right-hand side of the equation for $a$, $\psi$ is replaced by $\psi - \dot\chi y$ when the domain has zonal symmetry or by $\psi - \frac{1}{2}\dot\chi r^2$ when it has axial symmetry.

\bibliographystyle{alpha}
\bibliography{fot}

\end{document}